\documentclass[11pt]{article}
\usepackage[utf8]{inputenc}
\usepackage{amsmath}
\usepackage{lipsum} 
\usepackage{fullpage}
\usepackage{changepage}
\usepackage{mathrsfs}
\usepackage{svg}
\usepackage{xr}
\usepackage{amsthm,url,graphicx,booktabs,lipsum,color,bm,caption,subcaption,soul}
\usepackage{pifont,tikz,paralist,multirow,amssymb}
\usepackage{xifthen}
\usepackage{enumerate}
\usepackage{titlesec}
\usepackage{arydshln}
\usepackage{natbib} 
\usepackage{color} 
\usepackage{multirow}

\newtheorem{theorem}{Theorem}
\newtheorem{lemma}{Lemma} 

\newtheorem{proposition}{Proposition}
\newtheorem{assumption}{Assumption}
\newtheorem{corollary}{Corollary}
\newtheorem{condition}{Condition}
\usepackage{hyperref}
\hypersetup{
    colorlinks=true,
    linkcolor=blue,
    urlcolor=blue,
    citecolor=blue
}

\usepackage{threeparttable}
\usetikzlibrary{calc}
\usetikzlibrary{patterns}
\usetikzlibrary{positioning}
\allowdisplaybreaks

\makeatletter
\newcommand{\biggg}{\bBigg@{1.2}}

\makeatother

\def\exp{{\mathrm{exp}}}

\usetikzlibrary{positioning, shapes} 
\def\DTEone{{\scriptstyle{dte,1}}} 
\def\DTE{{\scriptstyle{dte,d}}} 
\def\dte{{\scriptstyle{dte}}} 
\def\DTEzero{{\scriptstyle{dte,0}}}  
\def\STE{{\scriptstyle{ste,d}}} 
\def\STEone{{\scriptstyle{ste,1}}} 
\def\ste{{\scriptstyle{ste}}}

\def\minK{{\scriptscriptstyle{K}}}

\def\STEzero{{\scriptstyle{ste,0}}} 
  
\def\ite{{\scriptstyle{ite}}} 
\def\ITE{{\scriptstyle{ite,d}}} 
\def\ITEone{{\scriptstyle{ite,1}}} 
\def\ITEzero{{\scriptstyle{ite,0}}}  
\def\eff{{\scriptstyle{\mathrm{eff}}}} 
\def\np{{\scriptstyle{{np}}}}

\def\T{{ \mathrm{\scriptscriptstyle T} }}
\makeatletter
\newcommand*{\indep}{%
 \mathbin{%
  \mathpalette{\@indep}{}%
 }%
}
\newcommand*{\nindep}{%
 \mathbin{
  \mathpalette{\@indep}{\not}
 }%
}
\newcommand*{\@indep}[2]{%
 \sbox0{$#1\perp\m@th$}
 \sbox2{$#1=$}
 \sbox4{$#1\vcenter{}$}
 \rlap{\copy0}
 \dimen@=\dimexpr\ht2-\ht4-.2pt\relax
 \kern\dimen@
 {#2}%
 \kern\dimen@
 \copy0 
}
 
\allowdisplaybreaks
\newcommand{\V}[1]{\boldsymbol{#1}} 

\makeatother

\def \P{\mathbb{P}} 
\def \V {\mathbb{V}}

\def \E {E}
\def\pr{\textnormal{pr}}

\usepackage{xr}
\def\pr{\textnormal{pr}}

\def\T{{ \mathrm{\scriptscriptstyle T} }}
\makeatletter
\newcommand*{\addFileDependency}[1]{
  \typeout{(#1)}
  \@addtofilelist{#1}
  \IfFileExists{#1}{}{\typeout{No file #1.}}
}
\makeatother

 \usepackage{tocloft}
 \usepackage{setspace}

\begin{document} 

\hypersetup{linkcolor=black}
\title{\bf \LARGE  Identification and estimation of causal peer
effects using instrumental variables}
\author{ Shanshan Luo\textsuperscript{1\dag}, Kang Shuai\textsuperscript{2\dag},     Yechi Zhang\textsuperscript{3}, 
 Wei Li\textsuperscript{3\thanks{Corresponding author}}, and Yangbo He\textsuperscript{2}   \\\\ 	\textsuperscript{1} School of Mathematics and Statistics, Beijing Technology and Business University\\\\\textsuperscript{2} School of Mathematical Sciences, Peking University\\\\ \textsuperscript{3} Center for Applied Statistics and School of Statistics,  Renmin University of China}

\date{} 
\maketitle  
\renewcommand{\thefootnote}{\dag}
\footnotetext{Equal contribution}
\hypersetup{linkcolor=blue}
\begin{abstract}
In social science researches, causal inference regarding peer effects often faces significant challenges due to homophily bias and contextual confounding. For example, unmeasured health conditions (e.g., influenza) and psychological states (e.g., happiness, loneliness) can spread among closely connected individuals, such as couples or siblings. To address these issues, we define four effect estimands for dyadic data to characterize {{direct effects and}}  spillover effects. We employ dual instrumental variables to achieve nonparametric identification of {{these causal estimands}} in the presence of unobserved confounding. We then derive the efficient influence functions for {{these estimands}} under the nonparametric model. Additionally, we develop a triply robust and locally efficient estimator that remains consistent even under partial misspecification of the observed data model.  The proposed robust estimators can be easily adapted to flexible approaches such as machine learning estimation methods,   provided that certain rate conditions are satisfied. Finally, we illustrate our approach through simulations and an empirical application evaluating the peer effects of retirement on fluid cognitive perception among couples.

\end{abstract}  
\section{Introduction}



Social and biomedical scientists have long been interested in how ego's behavior is affected by the behavior of its peers  \citep{meng2017impact}. There are many practical scenarios including group randomized experiments in school \citep{hong2006evaluating,li2019randomization}, twin studies or dyadic data in biology \citep{boomsma2002classical,voracek2007genetics,egami2023identification}, network experiments in economics \citep{manski1993identification, gao2023causal}, vaccine trials in epidemiology \citep{halloran1995causal,perez2014assessing,Park2023JASA} and policing assignments in city governance \citep{blattman2021place,savje2023causal}. Historically, causal inference has been primarily developed under potential outcome framework assuming the ego’s potential outcomes do not depend
on peer units’ treatments. Recently, causal inference researchers have focused on extending the classical potential outcome framework to accommodate interference \citep{hudgens2008toward,tchetgen2012causal,liu2014large}.



Although this issue is crucial, identifying and estimating causal peer effects face some challenges. It is often difficult to identify causal peer effects in observational studies due to unobserved network confounding factors, such as homophily bias and contextual confounding  \citep{manski1993identification,liu2022regression}. Homophily bias refers to the tendency for individuals to connect based on unobserved characteristics. Contextual confounding occurs when peers share some unobserved contextual factors. Concerns about these potential biases have led influential papers across interdisciplinary fields to criticize previous analyses of peer effects from observational studies  \citep{angrist2014perils}. \citet{shalizi2011homophily} discussed that it is nearly impossible to reliably estimate causal peer effects in observational studies using direct confounding adjustment methods (e.g., regression-based adjustment) due to pervasive concerns about unobserved network confounding.

In observational studies, two main approaches have been proposed to address concerns about uncontrolled network confounding. The first approach involves the use of negative control variables, which helps to detect and mitigate the impact of unmeasured confounding in network data. For instance, \citet{elwert2008wives} used negative control exposure to detect homophily bias in dyadic data but did not provide a formal method for leveraging negative control outcomes to fully account for this bias. Furthermore, \citet{egami2018identification} and \citet{liu2022regression} developed unbiased estimators for the average causal peer effect by establishing sufficient conditions for detecting unmeasured confounding using a single negative control variable. More recently, \citet{egami2023identification} explored the nonparametric identification of causal peer effects under network confounding using a pair of negative control variables for dyadic data.

The second approach involves instrumental variables (IVs), which mainly focuses on intention-to-treat (ITT) analysis and aims to estimate causal effects for complier subgroups rather than the entire population. Several studies have contributed to this approach. For instance, \citet{Vazquez2022spillover} analyzed the identification of direct effects and spillover effects under one-sided noncompliance, demonstrating that causal effects for certain subpopulations can be recovered using two-stage least square  estimators. Other works, such as those by \citet{kang2016peer}, \citet{kang2018spillover}, \citet{imai2021causal}, \citet{ditraglia2023identifying}, and \citet{hoshinoa2023causal}, also utilized IVs to estimate causal peer effects in the presence of confounding based on the ITT analysis. Additionally, \citet{bramoulle2009identification} and~\citet{o2014estimating} established the parametric identification of causal peer effects using an IV approach within linear models. But until now, there are still no researches on nonparametric identification of causal peer effect from the IV perspective.

In this paper, we formally define the average causal direct effect, average causal spillover effect, and average causal interaction effect within the dyadic context and IV framework, accounting for both interference and unmeasured network confounding. Specifically, we define the effect of the ego's treatment on the ego's outcome as the direct effect, and the effect of peer's treatment on ego's outcome as the spillover effect \citep{Vazquez2022spillover}. By applying a modified unmeasured common effect modifier assumption \citep{wang2018bounded,cui2024semiparametric}, we establish the nonparametric identification of the causal parameters of interest using dual IVs. We calculate the efficient influence functions (EIFs) for the causal parameters of interest, constructing triply robust estimators that are consistent across the union of three different observed data models and are locally semiparametric efficient. 



The rest of this paper is organized as follows. In Section \ref{sec2}, we introduce the IV framework for partial interference setting and define various causal estimands of interest. The sufficient conditions for establishing the identifiability of these effects are also provided. In Section \ref{sec:estimation}, we derive the EIF for these estimators and construct corresponding efficient estimators with triple robustness. We also discuss the flexible estimation of nuisance parameters, including machine learning  and  nonparametric estimation method. Simulation experiments and a real world data illustration are performed in Sections \ref{simulation} and \ref{sec:app}, respectively. We also consider an extension to 
the average interaction effect in Section \ref{sec:exten-AIE}. We conclude with a brief discussion in Section \ref{sec:discussion}. Proofs of all theoretical results are relegated to the Supplementary Material.

\section{Framework and notation}\label{sec2}

Suppose we observe $n$ independent and identically distributed dyadic samples (e.g., couples or twins), indexed by $i \in \{1, \ldots, n\}$.  For each pair of individuals in the sample, we denote $D_{i,k}$ as the treatment received by individual $k$, and let $Y_{i,k}$ denote the corresponding outcome variable, where $k \in \{1, 2\}$.  To clarify, we index $k=1$ as the ego unit and $k=2$ as the peer unit to assess the causal effects of interest. In our practical example, $D_{ i,1}$ and $D_{ i,2}$ represent whether the male and female are retired, respectively. In practice, identifying peer effects is often challenging due to the presence of homophily bias and contextual confounding, with $D_{i,1}$, $D_{i,2}$, $Y_{i,1}$, and $Y_{i,2}$ all subject to observed covariates $X_i$ and unobserved confounders $U_i$.   For simplicity, we omit the subscript $i$.
 \begin{figure}[t!]
\centering 
  \vspace{0.5cm} 
\begin{minipage}{0.3\textwidth}
\centering 
    \begin{tikzpicture}[
      node distance=1.5cm,
      every node/.style={circle, draw, minimum size=0.6cm},
      every path/.style={-stealth, thick}
    ]
      \node (D1) {$D_1$};
      \node (Y2) at ([yshift=-1cm] D1.south) {$Y_1$}; 
      \node (D2) at ([xshift=2cm] D1.east) {$D_2$}; 
      \node (Y1) at ([yshift=-1cm] D2.south) {$Y_2$}; 
      \node (U) at ([yshift=-1.25cm]$(Y1)!0.5!(Y2)$) [rectangle , minimum size=0.5cm] {$S=1$}; 
      \node (U1) at ([xshift=-1.5cm] U.west) [circle, draw, fill=gray!10, minimum size=0.5cm]  {$U_1$}; 
      \node (U2) at ([xshift=1.5cm] U.east)[circle, draw, fill=gray!10, minimum size=0.6cm] {$U_2$}; 
      
      \draw[line width=0.5pt] (D1) -- (Y1);
      \draw[line width=1.25pt] (D1) -- (Y2); 
      \draw[line width=0.5pt] (D2) -- (Y1);
      \draw[line width=1.25pt] (D2) -- (Y2); 
      \draw[line width=0.5pt] (U1) -- (Y2); 
      \draw[line width=0.5pt] (U2) -- (Y1); 
  \draw[line width=0.5pt] (U1) to [bend left=30] (D1);  
      \draw[line width=0.5pt] (U1) -- (U);
      \draw[line width=0.5pt] (U2) -- (U);
  \draw[line width=0.5pt] (U2) to [bend right=30] (D2);  
    \end{tikzpicture}  
    \centering (a)
\end{minipage} 
\hspace{0.1\textwidth}
\begin{minipage}{0.3\textwidth}
\centering
     \begin{tikzpicture}[
      node distance=1.5cm,
      every node/.style={circle, draw, minimum size=0.6cm},
      every path/.style={-stealth, thick}
    ]
      \node (D1) {$D_1$};
      \node (Y2) at ([yshift=-1cm] D1.south) {$Y_1$}; 
      \node (D2) at ([xshift=3cm] D1.east) {$D_2$}; 
      \node (Y1) at ([yshift=-1cm] D2.south) {$Y_2$}; 
  \node (U) at ([yshift=-1.25cm]$(Y1)!0.5!(Y2)$) [circle, draw, fill=gray!10, minimum size=0.6cm] {$U$}; 
      \draw[line width=0.5pt] (D1) -- (Y1);
      \draw[line width=1.25pt] (D1) -- (Y2); 
      \draw[line width=0.5pt] (D2) -- (Y1);
      \draw[line width=1.25pt] (D2) -- (Y2);
      \draw[line width=0.5pt] (U) -- (D1);
      \draw[line width=0.5pt] (U) -- (Y2);
      \draw[line width=0.5pt] (U) -- (D2);
      \draw[line width=0.5pt] (U) -- (Y1);
    \end{tikzpicture} 
    \centering (b)
\end{minipage}  
\caption{Directed acyclic graphs for dyadic data in the presence of unmeasured network confounding. (a) Unmeasured homophily. (b) Contextual confounding. The thick arrows from $(D_1,D_2)$ to  $Y_1$ indicate the causal effects of interest. The shaded nodes denote the unobserved confounders. For simplicity, we suppressed observed covariates $X$. In panel (a), the square box around $S = 1$ represents that we observe dyads conditional on $S = 1$.}
\label{fig:figures-11}
\end{figure}
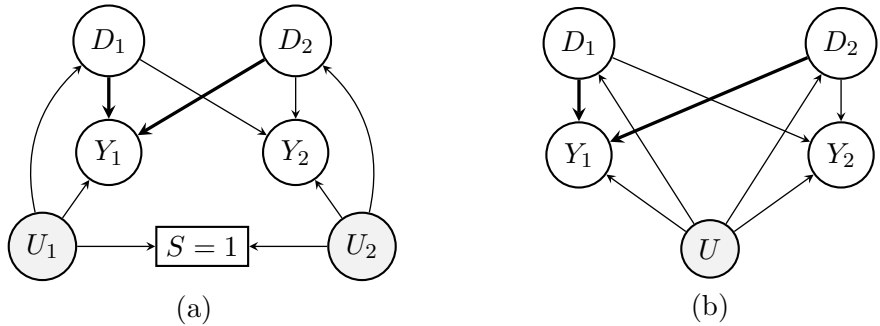

Figure \ref{fig:figures-11} presents two directed acyclic graphs (DAGs) illustrating potential sources of unmeasured confounding in dyadic data. In the left panel (a), we use a binary variable $S$ to denote whether there is a tie for the two units to characterize the homophily bias, where $U=(U_1,U_2)$ denotes the unobserved features. Since the bond relation $S$ is affected by $U_1$ and $U_2$, it is essentially a collider in graph-theoretic terms \citep{Pearl:2011}. When we condition on $S$, the spurious association {{among $\{ D_1,D_2,Y_1,Y_2\}$}} may lead to biased estimates and invalid conclusions{{~\citep{elwert2014endogenous}}}. In the right panel (b), the causal DAG illustrates unmeasured contextual confounders $U$, which simultaneously  affects both the ego and the peer through the pathways $Y_1 \leftarrow U \rightarrow D_1$ and  $Y_1 \leftarrow U \rightarrow D_2$.


Suppose the ego and the peer unit each have one binary IV, $Z_1 \in \{0,1\}$ and $Z_2 \in \{0,1\}$, where we assume the condition $Z_1 \indep Z_2 \mid X$. This requirement can be satisfied through randomized trials or experimental designs.  We next focus on causal estimands for the ego unit (i.e., unit 1) and the corresponding outcome $Y_1$, as the indices of the two units are essentially exchangeable. Let $O =(X ,Z_1,Z_2,D_1,D_2,Y_1)$ denote the vector of all observed variables. We first impose the following assumptions for two IVs $Z_1$ and $Z_2$.
  \begin{assumption}[Instrumental variables] 
 \label{assump:iv}
 Suppose the following assumptions hold:
  
   (a) Independence of the instruments: $ (Z_1,Z_2)\indep U\mid X $;

  (b) No instrument spillover: $(Z_1,D_1)\indep (Z_2,D_2) \mid (U,X)$;
 
    (c) Outcome exclusion restriction:  $(Z_1,Z_2)\indep Y_1\mid (D_1,D_2,U,X)$;

    (d) Relevance of the instruments:  $Z_1\nindep D_1\mid X $ and $Z_2\nindep D_2\mid X$.       
 \end{assumption}
Assumption \ref{assump:iv}(a) ensures that the two IVs $Z_1$ and $Z_2$ are unconfounded given $X$, which naturally holds under the randomized instrument settings. 
 Assumption \ref{assump:iv}(b) requires that $Z_1$ and $Z_2$ are both single-treatment IV, meaning that no causal effect of $Z_1$ on $D_2$ or $Z_2$ on $D_1$ exists. 
Assumption \ref{assump:iv}(c) implies that the two IVs $(Z_1,Z_2)$  do not directly affect the ego's outcome $Y_1$, so the potential outcomes are functions of treatment status only \citep{Vazquez2022spillover}. Let $Y_1(z_1, z_2, d_1, d_2)$ denote the potential value of outcome $Y_1$ when $Z_1,Z_2,D_1,D_2$ were respectively set to be $z_1,z_2,d_1,d_2$. Assumption \ref{assump:iv}(c) implies that $Y_1(z_1, z_2, d_1, d_2) = Y_1(d_1, d_2)$, thereby ensuring that two IVs $Z_1$ and $Z_2$ have no direct effect on $Y_1$. Assumption \ref{assump:iv}(d) is commonly adopted in instrumental variable analysis and can be verified with observed data \citep{Angrist:1996}. Figure \ref{fig:figures} completes Figure \ref{fig:figures-11} by incorporating the case where additional two IVs satisfying Assumption \ref{assump:iv} are present. Based on the simplified potential outcome notation from Assumption~\ref{assump:iv}, the following assumptions ensure the confounding sufficiency of $(X,U)$ and the positivity of two treatments and IVs. 

 \begin{assumption} 
 \label{assump:latent_overlap}
 Suppose the following assumptions hold:
 
   (a) {{Latent ignorability}}: $ (Z_1,Z_2,D_1,D_2)\indep 
   Y_1(d_1,d_2) \mid (X,U) $;

  (b) {{Overlap}}:   $0<\mathrm{pr}(D_1 =d_1, D_2=d_2 , Z_1 =z_1, Z_2 =z_2\mid X )<1$ for $z_1,z_2,d_1,d_2\in \{0,1\}$.  
 \end{assumption}
  Assumption~\ref{assump:latent_overlap}(a) implies that all potential confoundings are captured by the unmeasured confounders $ U $ and the baseline covariates $ X $. This assumption generalizes the standard ignorability condition \citep{rosenbaum1983central} to interference settings with unmeasured network confoundings.  Assumption~\ref{assump:latent_overlap}(b) implies that, for each level of measured  covariates 
 $X$, every possible combination of the two treatment variables and the two IVs has positive probability.  
 
Given the peer's treatment status $ d_2 $, the individual direct effect is defined as the effect of one's ego treatment on its ego outcome, namely, $ Y_{1}(1, d_2) - Y_{1}(0, d_2) $. Given one's ego treatment status $ d_1 $, the individual peer effect or spillover effect is defined as the effect of the peer's treatment on one's ego outcome, namely, $ Y_{1}(d_1, 1) - Y_{1}(d_1, 0) $. In this paper, we mainly focus on the following two estimands: the average direct effect $ \Delta_{\DTE}=\E\{Y_{1}(1, d) - Y_{1}(0, d)\} $ given the peer's treatment status $ d $, and the average spillover effect $ \Delta_{\STE}=\E\{Y_{1}(d, 1) - Y_{1}(d, 0)\} $ given the individual's ego treatment status $ d $.   In the following, the subscripts ``{\it dte}" and ``{\it ste}" are frequently used to refer to expressions related to the average direct treatment effect and the average spillover treatment effect, respectively. We now propose some novel conditions to ensure the nonparametric identification of the average causal direct and spillover effects. 
These conditions are motivated by researches using IVs to address unmeasured confounding  \citep{wang2018bounded,cui2021semiparametric}.

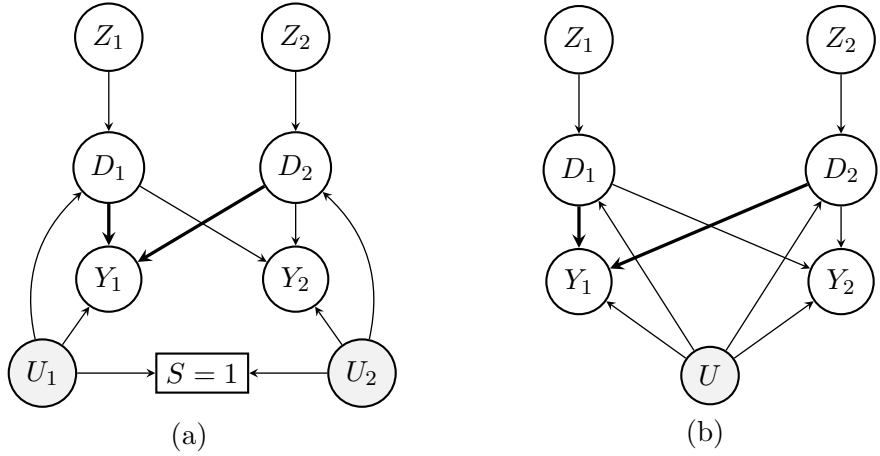
\begin{figure}[t]
\centering 
  \vspace{0.5cm} 
\begin{minipage}{0.3\textwidth}
\centering
    \begin{tikzpicture}[
      node distance=1.5cm,
      every node/.style={circle, draw, minimum size=0.6cm},
      every path/.style={-stealth, thick}
    ]
      \node (D1) {$D_1$};
      \node (Y2) at ([yshift=-1cm] D1.south) {$Y_1$}; 
      \node (D2) at ([xshift=2cm] D1.east) {$D_2$}; 
      \node (Y1) at ([yshift=-1cm] D2.south) {$Y_2$}; 
      \node (U) at ([yshift=-1.25cm]$(Y1)!0.5!(Y2)$) [rectangle , minimum size=0.5cm] {$S=1$}; 
      \node (U1) at ([xshift=-1.5cm] U.west) [circle, draw, fill=gray!10, minimum size=0.6cm]  {$U_1$}; 
      \node (U2) at ([xshift=1.5cm] U.east) [circle, draw, fill=gray!10, minimum size=0.6cm] {$U_2$}; 
 
      \draw[line width=0.5pt] (D1) -- (Y1);
      \draw[line width=1.25pt] (D1) -- (Y2); 
      \draw[line width=0.5pt] (D2) -- (Y1);
      \draw[line width=1.25pt] (D2) -- (Y2); 
      \draw[line width=0.5pt] (U1) -- (Y2); 
      \draw[line width=0.5pt] (U2) -- (Y1); 
      \node (Z1) at ([yshift=1.25cm] D1.north)   {$Z_1$}; 
      \node (Z2) at ([yshift=1.25cm] D2.north)   {$Z_2$};
      \draw[line width=0.5pt] (Z1) -- (D1); 
      \draw[line width=0.5pt] (Z2) -- (D2);  
      \draw[line width=0.5pt] (U1) to [bend left=30] (D1);  
      \draw[line width=0.5pt] (U1) -- (U);
      \draw[line width=0.5pt] (U2) -- (U);
      \draw[line width=0.5pt] (U2) to [bend right=30] (D2);  
    \end{tikzpicture}
    \centering (a)
    \label{fig:first}
\end{minipage} 
\hspace{0.1\textwidth}
\begin{minipage}{0.3\textwidth}
\centering
     \begin{tikzpicture}[
      node distance=1.5cm,
      every node/.style={circle, draw, minimum size=0.6cm},
      every path/.style={-stealth, thick}
    ]
      \node (D1) {$D_1$};
      \node (Y2) at ([yshift=-1cm] D1.south) {$Y_1$}; 
      \node (D2) at ([xshift=3cm] D1.east) {$D_2$}; 
      \node (Y1) at ([yshift=-1cm] D2.south) {$Y_2$}; 
      \node (U) at ([yshift=-1.25cm]$(Y1)!0.5!(Y2)$) [circle, draw, fill=gray!10, minimum size=0.6cm] {$U$}; 
      \node (Z1) at ([yshift=1.25cm] D1.north)   {$Z_1$}; 
      \node (Z2) at ([yshift=1.25cm] D2.north)   {$Z_2$};
      \draw[line width=0.5pt] (Z1) -- (D1); 
      \draw[line width=0.5pt] (Z2) -- (D2);  
      \draw[line width=0.5pt] (D1) -- (Y1);
      \draw[line width=1.25pt] (D1) -- (Y2); 
      \draw[line width=0.5pt] (D2) -- (Y1);
      \draw[line width=1.25pt] (D2) -- (Y2);
      \draw[line width=0.5pt] (U) -- (D1);
      \draw[line width=0.5pt] (U) -- (Y2);
      \draw[line width=0.5pt] (U) -- (D2);
      \draw[line width=0.5pt] (U) -- (Y1);
    \end{tikzpicture}
    \centering (b)
    \label{fig:second}
\end{minipage} 
\caption{Directed acyclic graphs for dyadic data in the presence of unmeasured network confounding. (a) Unmeasured homophily. (b) Contextual confounding. The arrows from $(D_1,D_2)$ to $Y_1$ indicate the causal effects of interest. We use shaded nodes to denote unobserved confounders.}
\label{fig:figures}
\end{figure}

\begin{assumption}[No unmeasured common effect modifier]
\label{assump:no-interaction} We consider the following conditions: 
 
(a) $ \operatorname{cov}\{\Gamma_{\DTE} (X, U), \delta_{ \DTE } (Z_2,X, U) \mid X\}=0,~~\text{for}~~d\in\{0,1\}, $

(b)  $\operatorname{cov}\{\Gamma_{\STE} (X, U), \delta_{ \STE } (Z_1,X, U) \mid X\}=0,~~\text{for}~~d\in\{0,1\},$\\
where 
{\small  \begin{gather*}
    \Gamma_{\DTE}  \left( X,U \right) = E\left\{Y_1\left( 1,d\right)-Y_1\left(0,d\right)\mid X,U\right\} ,~~\Gamma_{\STE}  \left(X,U\right) = E\left\{Y_1\left(  d,1\right)-Y_1\left(d,0\right)\mid X,U\right\}   ,\\ \delta_{{\DTE }} \left(Z_2, X,U\right) = E\left\{D_{1}\mathbb{I}(D_2=d) \mid Z_1=1 , Z_2,X,U\right\}  - E\left\{D_{1}\mathbb{I}(D_2=d)\mid Z_1=0 ,Z_2,X,U \right\} ,\\\delta_{{\STE}} \left( Z_1,X,U \right) = E\left\{\mathbb{I}(D_1=d)D_{2} \mid Z_1 , Z_2=1,X,U\right\}  -  E\left\{\mathbb{I}(D_1=d)D_{2} \mid Z_1 , Z_2=0,X,U\right\}  ,
\end{gather*}}
\end{assumption}
Due to the symmetric form of Assumption \ref{assump:no-interaction}(a) and \ref{assump:no-interaction}(b), we only provide some specific illustrations for Assumption \ref{assump:no-interaction}(a). Assumption \ref{assump:no-interaction}(a) essentially states that, conditioning on measured covariates $ X $, there is no common effect modifier by {{an}} unmeasured confounder for both (1) the additive effect of one's ego treatment on the ego outcome and (2) the additive effect of the peer's IV on the combination of two treatments $  D_1 $  and $ D_2 $. Taking $  d=1$  as an example, a more stringent version of Assumption \ref{assump:no-interaction}(a) can be stated as follows:  
\begin{gather*}
     ~~\Gamma_{\DTEone} (X, U) = E\left\{Y_1\left( 1,1\right)-Y_1\left(0,1\right)\mid X \right\} ,\\ \text{or }
~~\delta_{\DTEone} \left( Z_2,X, U\right) = E\left(D_{1}D_2 \mid Z_1=1 , Z_2,X\right) - E\left(D_{1}D_2 \mid Z_1=0 ,Z_2,X \right). 
\end{gather*}  
The first equality states that, when conditioning on measured covariates $ X $, unmeasured confounders $ U $ do not modify the causal effect of $ D_1 $ on $ Y_1 $ on the additive scale.  The second condition implies that, conditional on $ Z_2 $ and $ X $, unmeasured confounders $ U $ do not modify the causal effect of $ Z_1 $ on   $ D_1D_2 $ on the additive scale. Furthermore, in the absence of interference,  that is,   when $(Z_1, D_1, Y_1) \indep (Z_2, D_2) \mid (X,U)$, the potential outcomes reduce to $Y_1(d_1, d_2) = Y_1(d_1)$ for $d_1, d_2 \in \{0,1\}$. Consequently, the causal estimands involved in Assumption \ref{assump:no-interaction}  reduce to:
\begin{gather*} 
\Gamma_{\DTE} (X,U)  = E\left\{Y_1\left( 1 \right)-Y_1\left(0 \right)\mid X,U\right\} ,\\
 \Gamma_{\STE}  \left(X,U \right) = E\left\{Y_1\left(  d \right)-Y_1\left(d \right)\mid X,U\right\}   =0,\\\delta_{{\DTE}} \left( Z_2,X,U\right) =\left\{ E\left(D_{1}  \mid Z_1=1 , X,U\right) - E\left(D_{1}  \mid Z_1=0 , X,U \right)\right\}\pr(D_2=d\mid Z_2,X) ,\\\delta_{{\STE}} \left(Z_1,X,U\right) = \left\{E\left( D_2 \mid Z_2=1,X \right) - E\left( D_2 \mid  Z_2=0,X \right) \right\}\pr(D_1=d\mid Z_1,X,U). 
\end{gather*} 
 As a result, Assumption \ref{assump:no-interaction}(a) becomes $$\mathrm{cov} \big[ E\{ Y_1(1) - Y_1(0) \mid X,U \}, \{ E(D_1 \mid Z_1=1,X,U) - E(D_1 \mid Z_1 =0,X,U)  \} \mid X\big] = 0, $$ and Assumption \ref{assump:no-interaction}(b) naturally holds due to $\Gamma_{\STE}(Z_1,X,U)=0$. Therefore, this assumption can be considered as an extension of the unmeasured common effect modifier assumption by \citet{wang2018bounded}, adapted to the context of interference.

\begin{table}[t]
	\centering
  \vspace{0.5cm} 
	\caption{Summary of notations for direct effects used in this paper.}
	\label{tab:notations_direct}
 \resizebox{0.93\textwidth}{!}{
	\begin{tabular}{ccc}
		\toprule
    Notation & Definition  & Description   \\
        \midrule
		 $\pi_{\dte}(Z_1,X)$ & $\pr(Z_1\mid X)$    & Instrumental propensity score   \\
		$\delta_{\DTE}^j(X)$	& $\pr(D_{1}=1,D_2=d\mid Z_{{1}}=j,X)$    &  Conditional treatment probability  \\
		 $\delta_{\DTE}(X)$  & $\delta_{\DTE}^1(X) - \delta_{\DTE}^0(X)$   &  Treatment probability difference  \\
         $\mu_{\DTE}(X)$ & $E\{ D_{1} \mathbb{I}(D_{2}=d)\mid Z_1=0 ,X\}$   &  Conditional treatment expectation  \\
		 $\eta_{\DTE}(X)$ & $E\{ Y_1\mathbb{I}(D_{2}=d)\mid Z_1=0,X\}$  &  Conditional outcome expectation \\
       {{$\omega_{\DTE}(X)$}}  &  $E\{  Y_1(1,d) - Y_1(0,d) \mid X \}$  & Conditional effect expectation\\
		\bottomrule
	\end{tabular}%
 	}
\end{table}%

  To simplify the exposition, we introduce notations for the average direct effect in Table~\ref{tab:notations_direct}.  
Notations for the spillover effects are analogous and can be obtained by interchanging the indices of the two treatments and two instrumental variables.  
For example, \( \pi_{\ste}(Z_2, X) = \pr( Z_2 \mid X) \) represents the instrumental propensity score used for identifying spillover effects. These notations  for the spillover effects  are provided in Section~\ref{notation_spillover} of the Supplementary Material.  
As shown below, different identification formulas for the average direct effects can be derived from distinct functionals of the observed data.   


    \begin{theorem}
    \label{thm:iden}
    Under Assumptions \ref{assump:iv},\ref{assump:latent_overlap} and \ref{assump:no-interaction}(a), the average direct effect is identified by $  \Delta_{\DTE}  =E\left\{\omega_{\DTE}(X)\right\}$, where   \begin{align}
 \label{eq1}
 \omega_{\DTE} \left( X \right) = \frac{E\{ Y_1 \mathbb{I}(D_2 =d ) \mid Z_1 =1,X\} - E\{ Y_1 \mathbb{I}(D_2 =d)\mid Z_1 =0,X \} }{ E\{ D_1 \mathbb{I}(D_2 =d ) \mid Z_1 =1,X\} - E\{ D_1 \mathbb{I}(D_2 =d)\mid Z_1 =0,X \} }. 
 \end{align}
 Similarly, under Assumptions \ref{assump:iv},\ref{assump:latent_overlap} and \ref{assump:no-interaction}(b),  the average spillover effect   is identified by  $\Delta_{\STE} = E \left\{\omega_{\STE}(X) \right\}$, where  \begin{align}
 \label{eq2}
 \omega_{\STE} \left( X \right) = \frac{E\{ Y_1 \mathbb{I}(D_1 =d ) \mid Z_2 =1,X\} - E\{ Y_1 \mathbb{I}(D_1 =d)\mid Z_2 =0,X \} }{ E\{ D_2 \mathbb{I}(D_1 =d ) \mid Z_2 =1,X\} - E\{ D_2 \mathbb{I}(D_1 =d)\mid Z_2 =0,X \} }.
 \end{align}
\end{theorem}

Theorem \ref{thm:iden} presents our key identification result for $ \Delta_{\DTE} $ and $ \Delta_{\STE} $. Specifically, the nonparametric identification of the direct effect relies on the IV $Z_1$, while the identification of the spillover effect relies on the IV $Z_2$. The derived expressions take a ratio form of conditional expectations' difference, similar to the well-known Wald estimand in instrumental variable analysis, except that the outcome is multiplied by an indicator function. Interestingly, when there is indeed no spillover effect, $Z_2$ does not directly influence the term $Y_1 \mathbb{I}(D_1 = d)$, resulting in the second term \eqref{eq2} being equal to zero.     It is straightforward to observe that the identification expression of $\omega_{\DTE}(X)$ in~\eqref{eq1} of Theorem \ref{thm:iden} essentially corresponds to the conditional Wald estimand within the subset data with $D_2 =d$, namely,
     \begin{align*} 
 \omega_{\DTE} \left( X \right) = \frac{E( Y_1  \mid Z_1 =1,D_2 =d,X) - E( Y_1  \mid Z_1 =0,D_2 =d,X ) }{ E( D_1   \mid Z_1 =1,D_2 =d,X) - E( D_1  \mid Z_1 =0,D_2 =d,X) }. 
 \end{align*} Such identification expressions based on subset adjustment have also been discussed in many previous studies \citep{blackwell2017instrumental,Vazquez2022spillover,Blackwell03042023}. Specifically, the above expression shares similarities with a closely related paper in the interference context by~\citet{Vazquez2022spillover}, whose Proposition 1 similarly considers a Wald estimand on the subset   $Z_2 = z $ for $z\in\{0,1\}$, focusing on the local average treatment effect, while our goal is to recover the average effect over the entire population.  Additionally, in a related paper involving multiple treatment–instrument pairs, \citet{blackwell2017instrumental}  defined conditional effects and causal interaction effects within different principal strata, and their Theorem 1 similarly considers identifying these causal estimands using the subset where $ D_2 = d $ for $d\in\{0,1\}$.

From an initial perspective, Theorem~\ref{thm:iden} may appear to allow the identification of the average direct effect based on the subset where $D_2=d$, using some standard IV analyses similar to those established  in~\citet{wang2018bounded}. However, the identifiability result in that paper cannot be directly applied to the dyadic setting with unmeasured confounding, for the following three reasons.
First, \citet{wang2018bounded} did not account for interference in their framework, and thus their assumptions would need to be fully reformulated conditional on $D_2=d$ in order to be applicable.  
Second, even under their set of assumptions, their identification strategy can only establish the causal effect of $D_1$ on $Y_1$ within the subpopulation defined by $(D_2=d, X)$. Specifically, their method identifies quantities such as $\mathbb{E}\left\{ Y_1(D_1=1) - Y_1(D_1=0) \mid D_2=d, X \right\}$, where the potential outcome $Y_1(d_1)$ represents the value of $Y_1$ under treatment assignment $D_1=d_1$.  
Finally, although the conditional expression identified on the subset $(D_2=d, X)$ matches the form $\omega_{{\DTE}}(X)$ in our framework, after integrating over $X$, their approach can only recover the marginal causal estimand $\mathbb{E}\left\{ Y_1(D_1=1) - Y_1(D_1=0) \mid D_2=d \right\}$, which is not the target parameter of our analysis. As a result, the methods proposed in~\citet{wang2018bounded} cannot be directly applied, and it is necessary to develop new estimation or modeling approaches, as demonstrated later in our paper.

From now on, we mainly discuss the theoretical results regarding the direct effect estimand $ \Delta_{{\DTEone}}$, as similar derivations can be easily applied to other effect estimands. Three additional representations for $ \omega_{\DTEone}(X) $ in Theorem \ref{thm:iden} are presented in the following corollary, each involving a distinct set of  observed data functionals. 

\begin{corollary}
     \label{cor:alternative-form}   Under Assumptions \ref{assump:iv},\ref{assump:latent_overlap} and \ref{assump:no-interaction}(a), we have $ \Delta_{{\DTEone}} =E\{ \omega_{\DTEone}(X)\}$, where $ \omega_{\DTEone}(X)$ can be obtained from the following three ways:
  \vspace{1mm}

    (a) Explicit representation:  $0=E\{\varphi_{\mathrm{ipw}} (O;\pi_{\dte},\delta_{\DTEone} )\mid X\}-\omega_{\DTEone}(X)$, where
    \begin{align*}
\varphi_{\mathrm{ipw}} (O;\pi_{\dte},\delta_{\DTEone} ) =\frac{(-1)^{1-Z_1} D_2Y_1}{\pi_{\dte}(Z_1,X)\delta_{\DTEone}(X)}.
    \end{align*}

    (b) Implicit representation:  $0=E\{\varphi_{\mathrm{g}} (O;\pi_{\dte},\omega_{\DTEone} )\mid X\} $, where
        \begin{gather*}
\varphi_{\mathrm{g}} (O;\pi_{\dte},\omega_{\DTEone} )  = \frac{(-1)^{1-Z_1} D_2}{\pi_{\dte}(Z_1,X) }\left\{Y_1-D_1 
\omega_{\DTEone}(X)\right\}.
    \end{gather*}

     (c) Implicit representation:  $0=E\{\varphi_{\mathrm{reg}} (O; \mu_{\DTEone}, \eta_{\DTEone}, \omega_{\DTEone} )\mid Z_1,X\} $, where
        \begin{align*}
\varphi_{\mathrm{reg}} &(O;\mu_{\DTEone}, \eta_{\DTEone}, \omega_{\DTEone} )   = D_2\{Y_1 -D_1 \omega_{\DTEone}(X) \} -  \eta_{\DTEone}(X)  + \mu_{\DTEone}(X) \omega_{\DTEone}(X). 
    \end{align*}
\end{corollary}
Corollary \ref{cor:alternative-form} provides three distinct representations for the average direct effect parameter $ \Delta_{\DTEone} $, each leveraging different estimands represented as functionals of the observed data to establish identification formulas.  The first representation (a) is based on two inverse probabilities and can be regarded as an inverse probability weighting type estimator. The second representation (b) builds on the structure of the g-formula, incorporating both propensity scores and the conditional expectation  $\omega_{\DTEone}(X)$. The third representation (c) is entirely based on two  conditional expectations   $\omega_{\DTEone}(X)$ and $\mu_{\DTEone}(X)$.


\section{Estimation}
\label{sec:estimation}
\subsection{Semiparametric theory}
The multiple identification expressions for $ \Delta_{\DTEone} $ based on different components of the observed distribution, provided by Theorem \ref{thm:iden} and Corollary \ref{cor:alternative-form} in the previous section, naturally inspire the development of various estimation methods.  However, it also necessitates the proposal of a more principled estimation approach. We focus on deriving the EIF for $ \Delta_{\DTEone} $ in this section, which serves as a foundation for constructing efficient estimators \citep{bickel1993efficient}. 
Let $ \theta = \{\pi_{{\dte}}(X), \mu_{\DTEone}(X), \eta_{\DTEone}(X), \delta_{\DTEone}(X), \omega_{\DTEone}(X)\} $, where we sometimes omit $ X $ by treating each functional as a general nuisance parameter and use $ \theta = (\pi_{{\dte}}, \mu_{\DTEone}, \eta_{\DTEone}, \delta_{\DTEone}, \omega_{\DTEone}) $ for simplicity. 
Let \( \widehat{\theta} \) denote an estimator of \( \theta \), which may be derived from parametric models with finite-dimensional parameters,  flexible nonparametric approaches, or machine learning methods. A general estimation procedure for parametric models is provided in Section~\ref{sec:par-est}.
Any regular and asymptotically linear estimator $\widehat{\Delta} _{\DTEone}$ has one associated influence function $\varphi \left( O ;\theta \right) $ such that: 
$
\widehat{\Delta} _{\DTEone} - \Delta _{\DTEone} =\mathbb{P}_n\{\varphi \left( O ;\theta \right) \} + o_p(n^{-1/2})$, where $\mathbb{P}_n(\cdot)$ represents the sample averaging operator. 
We consider the nonparametric model $\mathcal{M}_{\text{np}}$ characterized by Assumptions \ref{assump:iv},\ref{assump:latent_overlap} and \ref{assump:no-interaction}(a):
\[
\mathcal{M}_{\text{np}} = \{ f(P) : \text{Assumptions \ref{assump:iv},\ref{assump:latent_overlap} and \ref{assump:no-interaction}(a) hold; } \pi_{{\dte}}, \mu_{\DTEone}, \eta_{\DTEone}, \delta_{\DTEone}, \omega_{\DTEone} \text{ are unrestricted} \}.
\]
Since we do not impose any restrictions on $ \mathcal{M}_{\text{np}} $ except for regularity conditions, the nonparametric estimators of $ \Delta_{\DTEone} $ based on Theorem \ref{thm:iden} and Corollary \ref{cor:alternative-form} are in fact asymptotically equivalent, with a common influence function given in the following Theorem \ref{thm:eif-2}. 
\begin{theorem}
\label{thm:eif-2}
   Under Assumptions \ref{assump:iv},\ref{assump:latent_overlap} and \ref{assump:no-interaction}(a), the EIF for     $\Delta _{\DTEone}$ in  $\mathcal{M}_{\mathrm{np}}$  is given by
\begin{gather*}
\varphi \left( O ; \theta \right) =
   \dfrac{ (-1)^{1-Z_1}}{\pi_{\dte}(Z_1,X) \delta_{\DTEone}(X) }\begin{Bmatrix}
       D_2 Y_1-  D_2D_1\omega_{\DTEone}(X)   \\\addlinespace[1mm]-\eta_{\DTEone}(X)+\mu_{\DTEone}(X)\omega_{\DTEone}(X)
    \end{Bmatrix}
+\omega_{\DTEone}(X) - \Delta_{\DTEone}.
\end{gather*} 
The semiparametric efficiency bound for $\Delta_{\DTEone}$ in  $\mathcal{M}_{\mathrm{np}}$ is $\V ^{\eff}=
E \{ \varphi \left( O ;\theta  \right) \}^2.$ 
\end{theorem}
The EIF form in Theorem \ref{thm:eif-2} motivates us to adopt the estimator $ \widehat{\Delta}_{\DTEone} = \mathbb{P}_n\{\varphi (O; \widehat{\theta})\} $, where the working models involved in $ {\theta} $ can be estimated using either parametric or machine learning approaches. A general estimation procedure for $ {\theta} $ is provided in Section \ref{sec:par-est} of the Supplementary Material. However, it is important to note that the working models in $ {\theta} $ include nuisance functions such as $ {\eta}_{\DTEone}(X) $ and $ {\delta}_{\DTEone}(X) $, which are challenging to correctly specify in practice. This requires us to consider the impact of different nuisance models on the asymptotic properties of the estimator $\widehat{\Delta}_{\DTEone}$.  
In the following subsections, we will explain that under certain constraints on the nuisance models, $ \widehat{\Delta}_{\DTEone} $ always remains consistent. Essentially, it is a triply robust estimator, ensuring valid inference   even if only a subset of the low-dimensional nuisance models are correctly specified. Additionally, we discuss certain constraints on the working models, allowing the use of nonparametric estimation methods to perform valid inference when interference is present.
\subsection{Multiply robust estimation}
\label{sec:mr-est}
In this section, we further discuss the EIF-based estimator $ \widehat{\Delta}_{\DTEone} $  and reveal its multiple robustness properties. We define the $ L_2 $-norm of a function $ h $ with respect to the distribution $ \nu $ of a generic variable $ V$ as $ \Vert h\Vert_2 = \{\int h^2(v)\nu(dv)\}^{1/2} $. Let $ \theta^\ast $ denote the probability limit of the nuisance estimator $ \widehat{\theta} $. In the following, we further provide the convergence rate requirements for the working models to ensure that the proposed estimator $\widehat{ \Delta }_{\DTEone}$ is consistent and achieves the semiparametric efficiency bound.  This concept is similar to the double debiased machine learning approach for estimating the average treatment effect \citep{chernozhukov2018double} and aligns with the principles of multiply robust estimation within other causal inference frameworks \citep{wang2018bounded,Negative-Shi,sun2022JMLR}.

  {{  \begin{assumption}
    \label{assump:regular}
        We consider the following conditions:


            (a) the functions   $\widehat\theta$ and $ \theta^\ast$ are in the same Donsker class;

            
            (b) there exist constants  $M>0$ and $0<\epsilon <1$ satisfying that:  
				\begin{gather*}
					\epsilon<\{ \widehat\pi_{\dte} ,\widehat \mu_{\DTEone} , \pi_{{\dte}},  \mu_{\DTEone} \} <1-\epsilon ,\\
					-M<  \{ {\widehat\eta}_{\DTEone} ,\widehat\omega_{\DTEone},{\widehat\delta}_{\DTEone}, \eta_{\DTEone} , \omega_{\DTEone}, \delta_{\DTEone} \} <M.
				\end{gather*}   
                
             (c)  the following three product terms  are   of order   $o_p(n^{-1/2})$,
			\begin{equation}    \label{eq: bias-terms}
    	\begin{gathered}
		 {\big\|\omega_{\DTEone}-{\widehat\omega}_{\DTEone}\big\|}_2{\big\|\pi_{\dte}\delta_{\DTEone}-\widehat\pi_{\dte}{\widehat\delta}_{\DTEone}\big\|}_2,\\{\big\|\pi_{\dte}-\widehat\pi_{\dte}\big\|}_2 {\big\|\mu_{\DTEone}-{\widehat\mu}_{{\DTEone}}\big\|}_2,~{\big\|\pi_{\dte}-\widehat\pi_{\dte}\big\|}_2{\big\|{{\widehat\eta}}_{\DTEone}-\eta_{\DTEone}\big\|}_2 .
				\end{gathered}
			\end{equation}
    \end{assumption}
}}
         Assumption \ref{assump:regular}(a) restricts the complexity of the function spaces, including both the nuisance functions and their limits \citep{kennedy2016semiparametric, van2000asymptotic}. While Donsker classes cover standard parametric families, they can also accommodate many infinite-dimensional function classes, provided they meet smoothness or boundedness conditions. To avoid stringent empirical process assumptions, techniques like sample splitting or cross-fitting may be used \citep{chernozhukov2018double}. 
The bounded constraints outlined in Assumption \ref{assump:regular}(b) are sufficient to control the error $ \widehat{\omega}_{\DTEone} - \omega_{\DTEone} $, with the terms summing up in forms analogous to \eqref{eq: bias-terms}. Similar to Assumption \ref{assump:regular}(a), Assumption \ref{assump:regular}(c) is relatively mild and holds in many scenarios. Specifically, satisfying condition (c) also ensures that the elements in each product term are $o_p(1)$. Furthermore, Assumption \ref{assump:regular}(c) is fulfilled if the nuisance function estimates converge at a rate faster than $ n^{-1/4} $. Under sparsity or smoothness assumptions, many flexible models can achieve this rate. Machine learning techniques, including random forests \citep{chernozhukov2022automatic}, can meet this requirement. Parametric methods, such as maximum likelihood estimation, can also achieve this rate, though they are more sensitive to model misspecification.  

{{\begin{theorem}
		\label{THM: NON-MARCHINE-LEARNING}	 
 Under Assumptions \ref{assump:iv},\ref{assump:latent_overlap}, \ref{assump:no-interaction}(a), and Assumption \ref{assump:regular}(a,b), the estimator $ \widehat{\Delta}_{\DTEone} $ is consistent with rate of convergence
\begin{equation}
\label{convergence_rate}
\begin{gathered}
\widehat{\Delta}_{\DTEone} - \Delta_{\DTEone} 
= 
O_p\begin{pmatrix}\addlinespace[0.25mm]
    n^{-1/2} + \lVert\omega_{\DTEone}-{\widehat\omega}_{\DTEone}\rVert_2 \lVert\pi_{\dte}\delta_{\DTEone}-\widehat\pi_{\dte}{\widehat\delta}_{\DTEone}\rVert_2 \\\addlinespace[1mm]+ \lVert\pi_{\dte}-\widehat\pi_{\dte}\rVert_2 \lVert\mu_{\DTEone}-
{\widehat\mu}_{\DTEone}\rVert_2
+\lVert\pi_{\dte}-\widehat\pi_{\dte}\rVert_2 \lVert{\widehat\eta}_{\DTEone}-\eta_{\DTEone}\rVert_2 \\\addlinespace[0.25mm]
\end{pmatrix}.  
\end{gathered}
\end{equation}
If Assumption~\ref{assump:regular}(c) also holds, $\widehat{\Delta}_{\DTEone}$ is asymptotically normal and semiparametric efficient, where
$$
\sqrt{n} (\widehat{\Delta}_{\DTEone} - {\Delta}_{\DTEone}) \overset{d}{\to} 
N \bigg( 0, E\big\{ \varphi(O; \theta) \big\}^2\bigg). 
$$
\end{theorem} 
}}

\begin{table}[h]
\centering 

\caption{Multiply robust estimation for various target causal effect estimands.}
\label{model: all-estimators}
\resizebox{0.995\columnwidth}{!}{%
\begin{threeparttable}
\begin{tabular}{ccccccc} 
\toprule
Estimand                                 &  & Multiply robust estimator                                                                                                                                                                                                                                                                                                                                                                            &                   & $\mathcal{M}_\text{ipw}$                                                                 & $\mathcal{M}_\text{g}$                                                                    & $\mathcal{M}_\text{reg}$  \\ 
\cline{3-3}\cline{5-7}\addlinespace[5mm]
\multirow{3}{*}{$\Delta_{\DTEone}$}  &  & \multirow{3}{*}{$ \begin{array}{c} \mathbb{P}_n\left[\begin{array}{c}\dfrac{(-1)^{1-Z_1}}{\widehat\pi_{\dte} (Z_1,X){\widehat\delta}_{\DTEone}(X)}\begin{Bmatrix} D_2Y_1- {\widehat\eta}_{\DTEone}(X)-D_1D_2\widehat\omega_{\DTEone}(X)\\\addlinespace[1mm]+ \widehat\mu_{\DTEone}(X) \widehat\omega_{\DTEone}(X)\end{Bmatrix}+\widehat\omega_{\DTEone}(X) \end{array}\right] \end{array}$}          & \multirow{3}{*}{} & \multirow{3}{*}{$ \begin{array}{c} \delta_{\DTEone}(X)\\\pi_{\dte}(Z_1,X) \end{array}$}  & \multirow{3}{*}{$ \begin{array}{c} \omega_{\DTEone}(X)\\\pi_{\dte}(Z_1,X) \end{array} $}  & $\mu_\DTEone(X)$          \\
                                     &  &                                                                                                                                                                                                                                                                                                                                                                                                      &                   &                                                                                          &                                                                                           & $\eta_\DTEone(X)$         \\
                                     &  &                                                                                                                                                                                                                                                                                                                                                                                                      &                   &                                                                                          &                                                                                           & $\omega_{\DTEone}(X)$       \\ 
\addlinespace[5mm]\cline{3-3}\cline{5-7}\addlinespace[5mm]
\multirow{3}{*}{$\Delta_{\DTEzero}$} &  & \multirow{3}{*}{$\begin{array}{c} \mathbb{P}_n\left[\begin{array}{c}\dfrac{(-1)^{1-Z_1}}{\widehat\pi_{\dte} (Z_1,X){\widehat\delta}_{\DTEzero}(X)}\left\{(1-D_2)Y_1- {\widehat\eta}_{\DTEzero}(X)-D_1(1-D_2)\widehat\omega_{\DTEzero}(X)\right.\\\left.+ \widehat\mu_{\DTEzero}(X)\widehat\omega_{\DTEzero}(X)\right\}+\widehat\omega_{\DTEzero}(X) \end{array}\right] \end{array}$}                 & \multirow{3}{*}{} & \multirow{3}{*}{$ \begin{array}{c} \delta_{\DTEzero}(X)\\\pi_{\dte}(Z_1,X) \end{array}$} & \multirow{3}{*}{$ \begin{array}{c} \omega_{\DTEzero}(X)\\\pi_{\dte}(Z_1,X) \end{array} $} & $\mu_\DTEzero(X)$         \\
                                     &  &                                                                                                                                                                                                                                                                                                                                                                                                      &                   &                                                                                          &                                                                                           & $\eta_\DTEzero(X)$        \\
                                     &  &                                                                                                                                                                                                                                                                                                                                                                                                      &                   &                                                                                          &                                                                                           & $\omega_{\DTEzero}(X)$      \\ 
\addlinespace[5mm]\cline{3-3}\cline{5-7}\addlinespace[5mm]\addlinespace[2mm]
\multirow{3}{*}{$\Delta_{\STEone}$}  &  & \multirow{3}{*}{$\begin{array}{c} \mathbb{P}_n\left[\begin{array}{c}\dfrac{(-1)^{1-Z_2}}{\widehat\pi_{ \ste}(Z_2,X){\widehat\delta}_{\STEone}(X)}\left\{D_1Y_1-{\widehat\eta}_{\STEone}(X)-D_1D_2\widehat\omega_{\STEone}(X)\right.\\\addlinespace[1mm]\left.+\widehat\mu_{\STEone}(X)\widehat\omega_{\STEone}(X)\right\}+\widehat\omega_{\STEone}(X) \end{array}\right] \end{array}$}               & \multirow{3}{*}{} & \multirow{3}{*}{$ \begin{array}{c} \delta_{\STEone}(X)\\\pi_{\ste}(Z_2,X) \end{array}$}  & \multirow{3}{*}{$ \begin{array}{c} \omega_{\STEone}(X)\\\pi_{\ste}(Z_2,X) \end{array} $}  & $\mu_\STEone(X)$          \\
                                     &  &                                                                                                                                                                                                                                                                                                                                                                                                      &                   &                                                                                          &                                                                                           & $\eta_\STEone(X)$         \\
                                     &  &                                                                                                                                                                                                                                                                                                                                                                                                      &                   &                                                                                          &                                                                                           & $\omega_{\STEone}(X)$       \\ 
\addlinespace[5mm]\addlinespace[2mm]\cline{3-3}\cline{5-7}\addlinespace[5mm]
\multirow{3}{*}{$\Delta_{\STEzero}$} &  & \multirow{3}{*}{$\begin{array}{c} \mathbb{P}_n\left[\begin{array}{c}\dfrac{(-1)^{1-Z_2}}{\widehat\pi_{ \ste}(Z_2,X){\widehat\delta}_{\STEzero}(X)}\left\{(1-D_1)Y_1-{\widehat\eta}_{\STEzero}(X)-(1-D_1)D_2\widehat\omega_{\STEzero}(X)\right.\\\addlinespace[1mm]\left.+\widehat\mu_{\STEzero}(X)\widehat\omega_{\STEzero}(X)\right\}+\widehat\omega_{\STEzero}(X) \end{array}\right] \end{array}$} & \multirow{3}{*}{} & \multirow{3}{*}{$ \begin{array}{c} \delta_{\STEzero}(X)\\\pi_{\ste}(Z_2,X) \end{array}$} & \multirow{3}{*}{$ \begin{array}{c} \omega_{\STEzero}(X)\\\pi_{\ste}(Z_2,X) \end{array} $} & $\mu_\STEzero(X)$         \\
                                     &  &                                                                                                                                                                                                                                                                                                                                                                                                      &                   &                                                                                          &                                                                                           & $\eta_\STEzero(X)$        \\
                                     &  &                                                                                                                                                                                                                                                                                                                                                                                                      &                   &                                                                                          &                                                                                           & $\omega_\STEzero(X)$      \\\addlinespace[5mm]
\bottomrule
\end{tabular}
 \end{threeparttable} }
\end{table}

Theorem \ref{THM: NON-MARCHINE-LEARNING} establishes the asymptotic convergence of the proposed estimator $ \widehat{\Delta}_{\DTEone} $ under Assumptions \ref{assump:regular}(a-c), and also demonstrates the multiple robustness of the estimator $ \widehat{\Delta}_{\DTEone} $. Table \ref{model: all-estimators} summarizes the multiple robustness properties underlying each estimator. For example, for the direct effect estimand $\Delta_\DTEone$, three submodels in Table \ref{model: all-estimators} are defined as:
\[
\mathcal{M}_{\text{ipw}} = \{ f(P) : \text{Assumptions \ref{assump:iv},\ref{assump:latent_overlap} and \ref{assump:no-interaction}(a) hold; } \pi_{{\dte}}, \delta_{\DTEone}  \text{ are correctly specified} \},
\]
\[
\mathcal{M}_{\text{g}} = \{ f(P) : \text{Assumptions \ref{assump:iv},\ref{assump:latent_overlap} and \ref{assump:no-interaction}(a) hold; } \pi_{{\dte}},  \omega_{\DTEone}  \text{ are correctly specified} \},
\]
\[
\mathcal{M}_{\text{reg}} = \{ f(P) : \text{Assumptions \ref{assump:iv},\ref{assump:latent_overlap} and \ref{assump:no-interaction}(a) hold; }  \mu_{\DTEone}, \eta_{\DTEone},\omega_{\DTEone}  \text{ are correctly specified} \}.
\]
Under model $\mathcal{M}_{\text{ipw}}$, we have   
\begin{align}
    \label{ipw:rate}
    {\big\|\pi_{\dte}-\widehat\pi_{\dte}\big\|}_2  = o_p(1) 
\quad \text{and} \quad  
{\big\|{{\widehat\delta}}_{\DTEone}-\delta_{\DTEone}\big\|}_2 = o_p(1).
\end{align}
 We thus can establish a consistent estimator for $\Delta_\DTEone$ using Corollary \ref{cor:alternative-form}(a). Similarly, the same arguments apply to the submodel $\mathcal{M}_{\text{g}}$ in conjunction with Corollary \ref{cor:alternative-form}(b), and to the submodel $\mathcal{M}_{\text{reg}}$ in conjunction with Corollary \ref{cor:alternative-form}(c).   
More importantly, \eqref{ipw:rate} essentially implies that \eqref{convergence_rate} is of order $ o_p(1) $, and Theorem \ref{THM: NON-MARCHINE-LEARNING} ensures that the EIF-based estimator $\widehat\Delta_\DTEone$ remains consistent. More generally, when at least one of $\mathcal{M}_{\text{ipw}}$, $\mathcal{M}_{\text{reg}}$, or $\mathcal{M}_{\text{g}}$ holds, $\widehat\Delta_\DTEone$ remains consistent.

Furthermore, under the additional conditions of Assumptions \ref{assump:regular}(c), the estimator $ \widehat{\Delta}_{\DTEone} $ is asymptotically normal, has the influence function $ \varphi(O; \theta) $, and achieves the semiparametric efficiency bound. However, we should also acknowledge that the conditions in Theorem \ref{THM: NON-MARCHINE-LEARNING} impose rate restrictions on some unnatural models, such as the outcome model $\eta_{\DTEone}(X)$ and the treatment assignment model $\delta_{\DTEone}(X)$. Even with flexible machine learning methods, these models are sometimes difficult to accurately estimate. Many studies have pointed out that misspecification of the outcome model or extreme weighting estimators can possibly lead to highly unstable performance \citep{Kang2007SS}. {{Therefore, in the next section, we also introduce a calibration method that nonparametrically estimates weighting functions by incrementally imposing moment restrictions \citep{ai2022simple}, providing a complement of the proposed multiply robust estimators.}}

 \subsection{Nonparametric estimation} 
 \label{ssec:non-est} 
In the previous section, it is mentioned that constructing the  estimator $\widehat\Delta_{\DTEone}$ often requires modeling certain unnatural nuisance functions, which may possibly lead to model misspecification bias. {{The identification expression in Corollary \ref{cor:alternative-form}(a) relies on two weighting functions, $\pi_{\dte}(Z_1, X)$ and $\delta_{\DTEone}(X)$. 
 This structure allows us to estimate the two inverse weights nonparametrically, ensuring that the proposed estimator achieves the desired efficiency \citep{ai2022simple}. Before formally introducing the estimation approach, we first present a lemma based on series moment restrictions to streamline the exposition.}}
 
\begin{lemma}
\label{thm:solution-unique}
	For any integrable function $v(X)$, the following moment conditions  
\begin{align} 
\label{eq: identify-sieve-1}
\begin{gathered}
E\{\phi(X) v(X)\}=E\left\{ (-1)^{1-Z_1}D_1D_2 \psi(Z_1  ,X) v(X)\right\},\\
    E\{Z _1\psi(Z_1  ,X) v(X)\}=E\{v(X)\}=E\{(1-Z_1)\psi(Z_1  ,X) v(X)\} , 
\end{gathered}\end{align}
simultaneously hold if and only if   $\psi(Z_1  ,X) = \pi_{\dte}^{-1} (Z_1,X)$ and $\phi(X)= \delta_{\DTEone}(X)$   almost surely. 
\end{lemma}

Lemma \ref{thm:solution-unique} outlines an approach to solving for $\delta_{\DTEone}(X)$ and $\pi_{\dte}(Z_1, X)$ using \eqref{eq: identify-sieve-1} for any integrable function $v(X)$.
However, this approach is practically infeasible because the moment conditions inherently involve infinite constraints. To address this, we employ a sequence of basis functions $v_{\minK}(X)$. We expect that the linear combination of these basis functions can well approximate any square integrable real-valued function \citep{CHEN20075549}, thereby providing a practical approximation to \eqref{eq: identify-sieve-1}. The moment conditions for the sieve version $v_{\minK}(X)$ in \eqref{eq: identify-sieve-1} are:
\begin{align}
\label{eq:Z-sieve}
\begin{gathered}
 E\{\phi(X) v_{\minK}(X)\}=E\{(-1)^{1-Z_1}   D_1D_2\psi(Z_1  ,X) v_{\minK}(X)\},\\
    E\{Z_1  \psi(Z_1  ,X) v_{\minK}(X)\} =E\{v_{\minK}(X)\}=E\{(1-Z_1)\psi(Z_1  ,X) v_{\minK}(X)\}, 
\end{gathered}\end{align}
Since the sieve space is a subspace of the original functional space, the solutions to \eqref{eq: identify-sieve-1}, namely $\psi(Z_1, X) = \pi^{-1}_{{dte}}(Z_1,X)$ and $\phi(X) = \delta_{dte,1}(X)$, remain valid solutions to \eqref{eq:Z-sieve}. However, they may not be the unique solutions to \eqref{eq:Z-sieve}. 
To address this problem, we adopt approaches similar to \citet{ai2022simple}, utilizing two globally concave and increasing functions to solve \eqref{eq:Z-sieve}. Specifically, we consider the concave functions $m_1(v) = v - \exp(-v)$ and $m_2(v) = -\log(e^v + e^{-v})$ for any $v$. 
We hence  consider the following objective functions:
\begin{gather*}
H_1(\alpha)=E\left\{m_2\left(\alpha^{\T} v_{\minK}(X)\right)-  (-1)^{1-Z_1} D_1D_2 \psi_{ \minK}(Z_1 ,X)  \alpha^{\T} v_{\minK}(X)\right\},\\
H_2(\beta, \gamma)=E\left\{Z_1 m_1\left(\beta^{\T} v_{\minK}(X)\right)-\beta^{\T} v_{\minK}(X)\right\}+E\left\{(1-Z_1) m_1\left(\gamma^{\T} v_{\minK}(X)\right)-\gamma^{\T} v_{\minK}(X)\right\}.
\end{gather*} 
Let    $\alpha_{\minK}$ and $(\beta_{\minK}^\T, \gamma_{\minK}^\T)^\T$ represent the maximizers of the objective functions $H_1(\alpha)$ and $H_2(\beta, \gamma)$, respectively. 
We define $\psi_{ \minK}(Z_1 ,X)=Z_1  \dot{m}_1\left(\beta_{\minK}^{\T} v_{\minK}(X)\right)+(1-Z_1 ) \dot{m}_1\left(\gamma_{\minK}^{\T} v_{\minK}(X)\right)$ and $\phi_{ \minK}(X)=\dot{m}_2\left(\alpha_{\minK}^{\T} v_{\minK}(X)\right)$, where the $\cdot$ above the function represents the derivative of the corresponding function. Here, $\left\{\psi_{ \minK}(Z_1 ,X), \phi_{ \minK}(X)\right\}$ can be regarded as the solution to \eqref{eq:Z-sieve}, which is expected to converge to $\{\pi^{-1}_{\dte}(Z_1, X),\delta_{\DTEone}(X)\}$. Now we summarize the above approach in  finite samples.
\begin{itemize} 
    \item[{\it Step 1}.]  
    Compute  the estimators $\widehat{\alpha}_{\minK} = \arg \max_{\alpha} \widehat{H}_1(\alpha)$ and $(\widehat{\beta}_{\minK}, \widehat{\gamma}_{\minK}) = \arg \max_{\beta, \gamma} \widehat{H}_2(\beta, \gamma)$, where $\widehat{H}_1(\alpha)$ and $\widehat{H}_2(\beta, \gamma)$  are the sample versions of $H_1(\alpha)$ and $H_2(\beta, \gamma)$, respectively.

    \item[{\it Step 2}.]  
    Construct the estimators   $
    \widehat{\psi}(Z_1 ,X) = Z_1 \dot{m}_1(\widehat{\beta}_{\minK}^{\T} v_{\minK}(X)) + (1-Z_1) \dot{m}_1\left(\widehat{\gamma}_{\minK}^{\T} v_{\minK}(X)\right)$ and $
    \widehat{\phi}(X) = \dot{m}_2\left(\widehat{\alpha}_{\minK}^{\T} v_{\minK}(X)\right).$ 

    \item[{\it Step 3}.]  
    Obtain the nonparametric estimator for ${\Delta}_{\DTEone}$ as follows:  
\begin{equation}
    \label{est:np-estimator}
     \widehat{\Delta}_{\DTEone}^\np = \mathbb{P}_n\left\{\frac{(-1)^{1-Z_1} D_2 Y_1 \widehat{\psi}(Z_1 ,X)}{\widehat{\phi}(X)}\right\}.
\end{equation}
\end{itemize}
The following proposition demonstrates that  the estimator $ \widehat{\Delta}_{\DTEone}^\np $ can achieve the semiparametric efficiency bound proposed in Theorem \ref{thm:eif-2}. 
\begin{proposition}
\label{eq:asymp-var}
    If Assumptions  \ref{assump:iv},\ref{assump:latent_overlap},\ref{assump:no-interaction}(a) and Condition~\ref{assumption:regular condition} in the Supplementary Material are satisfied, we obtain: (i) $\widehat{\Delta}_{\DTEone}^\np \xrightarrow{p} {\Delta}_{\DTEone}$ and (ii) $\sqrt{n}(\widehat{\Delta}_{\DTEone}^\np-{\Delta}_{\DTEone}) \xrightarrow{d} N(0, E\big\{ \varphi(O; \theta) \big\}^2)$, where $\varphi(O; \theta)$ is the EIF derived by Theorem \ref{thm:eif-2}.
\end{proposition}

 \section{Simulation studies}
 \label{simulation}
  \subsection{Simulation settings}
 In this section, we conduct several simulation studies to evaluate the ﬁnite sample performance of the proposed estimators.  
The observed confounders $X=(X_1,X_2)$ are generated from the uniform distribution on $ [-1,1]^2$. The unmeasured confounders $U$ are generated from the uniform distribution $(0,0.5]^2$. IVs $Z=(Z_1,Z_2)$ are generated from the logistic models:
\begin{align*}
    \pr(Z_1\mid X)=\mathrm{expit}(0.25X_1+0.25X_2),~~
    \pr(Z_2\mid X)=\mathrm{expit}(0.25X_1+0.25X_2),
\end{align*}
where $\mathrm{expit}(v)=\exp(v)/\{1+\exp(v)\}$.  
Treatment variables $D=(D_1,D_2)$ are generated from the  multinomial   distribution:
\begin{gather*}
    \pr(D_1=1, D_2=1 \mid Z,X,U )=m(Z_{1},X,U) m(Z_{2},X,U)\\\pr(D_1=0, D_2=1\mid  Z,X,U\}=\{1-m(Z_{1},X,U)\}  m(Z_{2},X,U)  ,\\
 \pr(D_1=1,D_2=0\mid  Z,X,U)=m(Z_{1},X,U)\{{1-m(Z_{2},X,U)}\}  ,\\\pr(D_1=0,D_2=0\mid  Z,X,U)=\{1-m(Z_{1},X,U)\}  \{{1-m(Z_{2},X,U)}\}   , 
\end{gather*}
where $m(Z_j,X,U)= \mathrm{expit}(-1+2Z_j-0.25X_1-0.25X_2+0.05U_1 -0.05U_2)$ for $Z_j\in\{0,1\}$.  
Four potential outcomes are respectively generated from the linear models:
\begin{gather*}
  Y_1(1,1)=6+6X_1+5X_2+2U_1+2U_2 +\epsilon_{11} , ~Y_1(1,0)=3+4X_1+2X_2+2U_1+2U_2 +\epsilon_{10} ,\\
  Y_1(0,1)=-1+2X_1+1.5X_2+2U_1+2U_2 +\epsilon_{01} ,  ~Y_1(0,0)=-2+ X_1+0.5X_2+2U_1+2U_2 +\epsilon_{00} ,  
\end{gather*}
where the error term $\epsilon_{ij}$ follows the standard normal distribution ($i,j\in\{0,1\}$). The observed  outcome $Y_1$  is generated from: 
$$Y_1=D_1D_2Y_1(1,1)+D_1(1-D_2)Y_1(1,0)+(1-D_1)D_2 Y_1(0,1)+(1-D_1)(1-D_2)Y_1(0,0).$$
 The true values of $\Delta_\DTEone$, $\Delta_\DTEzero$, $\Delta_\STEone$,  and  $\Delta_\STEzero$  are 7, 5, 3, and 1, respectively. 
 \subsection{Estimation methods}

In the previous subsection, we find that \( \pi_{\dte}(X) \) and \( \pi_{\ste}(X) \) follow logistic models, while \( \omega_{\DTE}(X) \) and \( \omega_{\STE}(X) \) follow linear models for $d\in\{0,1\}$. These functions can therefore be correctly specified using standard parametric forms.
 In contrast, the other nuisance functions are more difficult to specify using conventional parametric models.  To assess the performance of the multiply robust estimators in Section~\ref{sec:mr-est} and the nonparametric estimators in Section~\ref{ssec:non-est}, we next consider a range of estimation strategies for the nuisance models. For simplicity, we only describe the working models used for estimating \( \widehat{\Delta}_{\DTEone} \), as the procedures for the remaining components are analogous. 
\begin{itemize}
    \item[(a)]
The estimator \( \widehat\Delta_{\DTEone} = \mathbb{P}_n\{\varphi(O; \widehat\theta)\} \) is used for inference, where \( \widehat{\theta} \) is obtained from parametric models. Specifically, \( \widehat{\pi}_{\dte} \) and \( \widehat\mu_{\DTEone} \) are estimated via logistic regression, \( \widehat{\eta}_{\DTEone} \) and \( \widehat{\omega}_{\DTEone} \) are modeled using linear regression, and \( \widehat{\delta}_{\DTEone} \) is modeled using the hyperbolic tangent function 
Accordingly, when implementing the multiply robust estimator, we specified the model \( \mathcal{M}_{\text{g}} \) correctly, while deliberately misspecifying \( \mathcal{M}_{\text{ipw}} \) and \( \mathcal{M}_{\text{reg}} \). This allows us to evaluate the sensitivity of the proposed method to model misspecification. Detailed estimation procedures are provided in Section~\ref{sec:par-est} of the Supplementary Material.

    \item[(b)] The estimator $\widehat\Delta_{\DTEone} = \mathbb{P}_n\{\varphi(O; \widehat\theta)\}$ is used for estimation, where $\widehat{\theta} $ is estimated using neural networks \citep{ripley1996pattern, venables2013modern}.  Specifically,   \( \widehat{\pi}_{\dte} \), \( \widehat{\mu}_{\DTEone} \), and \( \widehat{\eta}_{\DTEone} \) are estimated using a neural network with a single hidden layer containing 4 units, trained for a maximum of 500 iterations using the R package \texttt{nnet}.  
In addition, \( \widehat{\omega}_{\DTEone} \) is modeled by a linear model, and \( \widehat{\delta}_{\DTEone} \) is modeled using the hyperbolic tangent function, with estimation following the Steps 4 and 5 described in Section~\ref{sec:par-est}.

    \item[(c)] The estimator $\widehat\Delta_{\DTEone} = \mathbb{P}_n\{\varphi(O; \widehat\theta)\}$ is used for estimation, which $\widehat{\theta} $ are estimated using gradient boosting machines \citep{friedman2001greedy, greenwell2019gbm}. Specifically, \( \widehat{\pi}_{\dte} \), \( \widehat{\mu}_{\DTEone} \), and \( \widehat{\eta}_{\DTEone} \) are estimated using gradient boosting machines with 500 trees, a learning rate of 0.01, and a minimum of 10 observations per node, implemented via the R package \texttt{gbm}. The  estimation procedures for the remaining nuisance functions are similar to those used in the \texttt{nnet}-based implementation.

    \item[(d)] The estimator $\widehat\Delta_{\DTEone} = \mathbb{P}_n\{\varphi(O; \widehat\theta)\}$ is used for estimation, which $\widehat{\theta} $ are estimated using lasso regression \citep{tibshirani1996regression, friedman2010regularization}. Specifically, \( \widehat{\pi}_{\dte} \), \( \widehat{\mu}_{\DTEone} \), and \( \widehat{\eta}_{\DTEone} \) are estimated using lasso regression (\( L_1 \) regularization), with the regularization parameter selected via cross-validation using the `cv.glmnet` function from the R package \texttt{glmnet}. The  estimation procedures for the remaining nuisance functions are similar to those used in the \texttt{nnet}-based implementation.

    \item[(e)] The nonparametric sieve estimator $\widehat{\Delta}_{\DTEone}^\np$, as described in \eqref{est:np-estimator}, is used for estimation.
\end{itemize} 

  \subsection{Estimation results}
Table \ref{tab:sim-results} presents the simulation results for four causal peer effects using various estimation methods across different sample sizes. Figure \ref{fig:main} provides the corresponding boxplot results. All estimation results are based on 200 replications and 200 bootstrap samples.  
First, we observe that, except for \texttt{gbm}, all other estimators exhibit consistency across all sample sizes, with the   standard deviation decreasing as the sample size increases. Additionally, \texttt{gbm} shows biased estimates for certain estimands, such as $\Delta_\DTEone$ and $\Delta_\STEone$, even with larger sample sizes. 
Second, from the perspective of coverage probabilities, parametric models, \texttt{nnet}, and lasso methods demonstrate acceptable significance levels, while \texttt{gbm} and \texttt{sieve} do not exhibit significance. Moreover, as expected, the multiply robust estimators yield consistent estimates, even when some nuisance models are misspecified or estimated using machine learning methods.
Finally, nonparametric methods like \texttt{sieve} display unstable performance, with larger estimation biases and higher standard deviations, particularly for smaller sample sizes. However, their performance stabilizes as the sample size increases. In summary,  among these methods, the multiply robust estimators  $\widehat\Delta_{\DTEone} $ based on parametric models and \texttt{nnet} generally perform the best across all scenarios. 


  \begin{table}[ht]
\centering
    \caption{ Simulation results of  causal peer effects using various estimation methods, reporting the bias, standard deviation (SD), and 95\% coverage probability (CP) for each method.}
    \label{tab:sim-results}
\resizebox{0.9978489594986795\textwidth}{!}{\begin{tabular}{ccccccccccccccccccccc}
\toprule
   $\Delta_\DTEone$     &  & \multicolumn{3}{c}{Parametric} &  & \multicolumn{3}{c}{nnet} &  & \multicolumn{3}{c}{gbm} &  & \multicolumn{3}{c}{lasso} &  & \multicolumn{3}{c}{sieve} \\   \cline{1-1}   \cline{3-5} \cline{7-9} \cline{11-13} \cline{15-17} \cline{19-21} \addlinespace[1mm]
$n$    &  & Bias    & SD     & CP    &  & Bias    & SD     & CP    &  & Bias    & SD    & CP    &  & Bias    & SD     & CP     &  & Bias    & SD      & CP    \\
$1000$  &  & -0.01   & 0.36   & 1.00  &  & -0.09   & 0.95   & 1.00  &  & 0.12    & 0.38  & 1.00  &  & 0.06    & 0.38   & 1.00   &  & 3.01    & 36.89   & 0.97  \\
$2000$  &  & 0.02    & 0.25   & 0.96  &  & -0.09   & 1.03   & 0.97  &  & 0.16    & 0.25  & 0.94  &  & 0.06    & 0.25   & 0.95   &  & -0.26   & 14.86   & 0.95  \\
$5000$  &  & 0.01    & 0.15   & 0.95  &  & -0.01   & 0.27   & 0.96  &  & 0.14    & 0.15  & 0.88  &  & 0.03    & 0.15   & 0.95   &  & 0.10    & 0.90    & 1.00  \\
$10000$ &  & -0.01   & 0.10   & 0.96  &  & -0.01   & 0.10   & 0.96  &  & 0.14    & 0.10  & 0.74  &  & 0.01    & 0.10   & 0.96   &  & 0.01    & 9.13    & 1.00  \\$20000$ &   & 0.00 & 0.06 & 0.96 &  & 0.00 & 0.06 & 0.96 &  & 0.16 & 0.07 & 0.37 &  & 0.02 & 0.06 & 0.94 &  & 0.00 & 0.07 & 0.98 \\  \toprule
    $\Delta_\DTEzero$         &  & \multicolumn{3}{c}{Parametric} &  & \multicolumn{3}{c}{nnet} &  & \multicolumn{3}{c}{gbm} &  & \multicolumn{3}{c}{lasso} &  & \multicolumn{3}{c}{sieve} \\   \cline{1-1}\cline{3-5} \cline{7-9} \cline{11-13} \cline{15-17} \cline{19-21}  \addlinespace[1mm]
$n$     &  & Bias    & SD     & CP    &  & Bias    & SD     & CP    &  & Bias    & SD    & CP    &  & Bias    & SD     & CP     &  & Bias    & SD      & CP    \\
$1000$  &  & 0.46    & 5.77   & 0.98  &  & 0.23    & 7.09   & 0.98  &  & 0.17    & 0.91  & 0.98  &  & 0.06    & 0.61   & 0.98   &  & 0.09    & 27.04   & 0.94  \\
$2000$  &  & 0.04    & 0.22   & 0.98  &  & -0.11   & 0.92   & 0.97  &  & 0.08    & 0.22  & 0.98  &  & 0.06    & 0.22   & 0.98   &  & -1.32   & 20.36   & 0.97  \\
$5000$  &  & -0.01   & 0.13   & 0.94  &  & -0.01   & 0.13   & 0.95  &  & 0.04    & 0.13  & 0.97  &  & 0.01    & 0.13   & 0.94   &  & 0.04    & 1.74    & 0.99  \\
$10000$ &  & 0.01    & 0.09   & 0.96  &  & 0.01    & 0.09   & 0.96  &  & 0.07    & 0.09  & 0.94  &  & 0.02    & 0.09   & 0.95   &  & 0.01    & 0.20    & 0.98  \\ $20000$ &  &  0.01 & 0.06 & 0.95 &  & 0.01 & 0.06 & 0.96 &  & 0.07 & 0.06 & 0.81 &  & 0.02 & 0.06 & 0.95 &  & 0.01 & 0.06 & 0.98 \\  \toprule
     $\Delta_\STEone$    &  & \multicolumn{3}{c}{Parametric} &  & \multicolumn{3}{c}{nnet} &  & \multicolumn{3}{c}{gbm} &  & \multicolumn{3}{c}{lasso} &  & \multicolumn{3}{c}{sieve} \\ \cline{1-1} \cline{3-5} \cline{7-9} \cline{11-13} \cline{15-17} \cline{19-21}  \addlinespace[1mm]
$n$     &  & Bias    & SD     & CP    &  & Bias    & SD     & CP    &  & Bias    & SD    & CP    &  & Bias    & SD     & CP     &  & Bias    & SD      & CP    \\
$1000$  &  & -0.14   & 0.78   & 0.98  &  & -0.21   & 1.05   & 0.95  &  & -0.03   & 1.45  & 0.97  &  & -0.22   & 2.26   & 0.97   &  & 1.68    & 100.89  & 0.93  \\
$2000$  &  & -0.09   & 0.56   & 0.96  &  & -0.13   & 0.67   & 0.92  &  & 0.08    & 0.54  & 0.92  &  & -0.07   & 0.55   & 0.96   &  & 1.78    & 13.26   & 0.96  \\
$5000$  &  & -0.08   & 0.32   & 0.94  &  & -0.08   & 0.32   & 0.96  &  & 0.11    & 0.31  & 0.90  &  & -0.06   & 0.31   & 0.95   &  & -0.07   & 3.40    & 1.00  \\
$10000$ &  & -0.02   & 0.22   & 0.94  &  & -0.02   & 0.22   & 0.96  &  & 0.18    & 0.22  & 0.83  &  & -0.01   & 0.22   & 0.94   &  & -0.05   & 0.67    & 0.98  \\ 
$20000$ &   & -0.01 & 0.17 & 0.92 &  & -0.01 & 0.17 & 0.92 &  & 0.21 & 0.16 & 0.70 &  & 0.00 & 0.17 & 0.92 &  & -0.03 & 0.17 & 0.98 \\ \toprule
     
    $\Delta_\STEzero$       &  & \multicolumn{3}{c}{Parametric} &  & \multicolumn{3}{c}{nnet} &  & \multicolumn{3}{c}{gbm} &  & \multicolumn{3}{c}{lasso} &  & \multicolumn{3}{c}{sieve} \\  \cline{1-1}\cline{3-5} \cline{7-9} \cline{11-13} \cline{15-17} \cline{19-21}  \addlinespace[1mm]
$n$     &  & Bias    & SD     & CP    &  & Bias    & SD     & CP    &  & Bias    & SD    & CP    &  & Bias    & SD     & CP     &  & Bias    & SD      & CP    \\
$1000$  &  & 0.04    & 0.31   & 0.98  &  & 0.04    & 0.35   & 1.00  &  & 0.06    & 0.31  & 0.99  &  & 0.05    & 0.32   & 0.98   &  & 0.51    & 11.68   & 0.94  \\
$2000$  &  & 0.02    & 0.21   & 0.96  &  & 0.02    & 0.21   & 0.96  &  & 0.05    & 0.21  & 0.97  &  & 0.02    & 0.21   & 0.96   &  & 0.37    & 3.10    & 0.96  \\
$5000$  &  & 0.00   & 0.12   & 0.96  &  & 0.00   & 0.13   & 0.96  &  & 0.04    & 0.13  & 0.96  &  & 0.00    & 0.12   & 0.97   &  & -0.21   & 5.00    & 0.98  \\
$10000$ &  & 0.01    & 0.09   & 0.98  &  & 0.01    & 0.09   & 0.97  &  & 0.06    & 0.09  & 0.94  &  & 0.02    & 0.09   & 0.97   &  & 0.02    & 0.16    & 1.00  \\ 
 $20000$ &&  0.00 & 0.06 & 0.94 &  & 0.00 & 0.06 & 0.96 &  & 0.06 & 0.06 & 0.86 &  & 0.01 & 0.06 & 0.94 &  & 0.01 & 0.09 & 0.98 \\ \toprule
\end{tabular}}
\end{table}
\begin{figure}
    \centering
    \includegraphics[width=\textwidth, height=0.95\textheight, keepaspectratio]{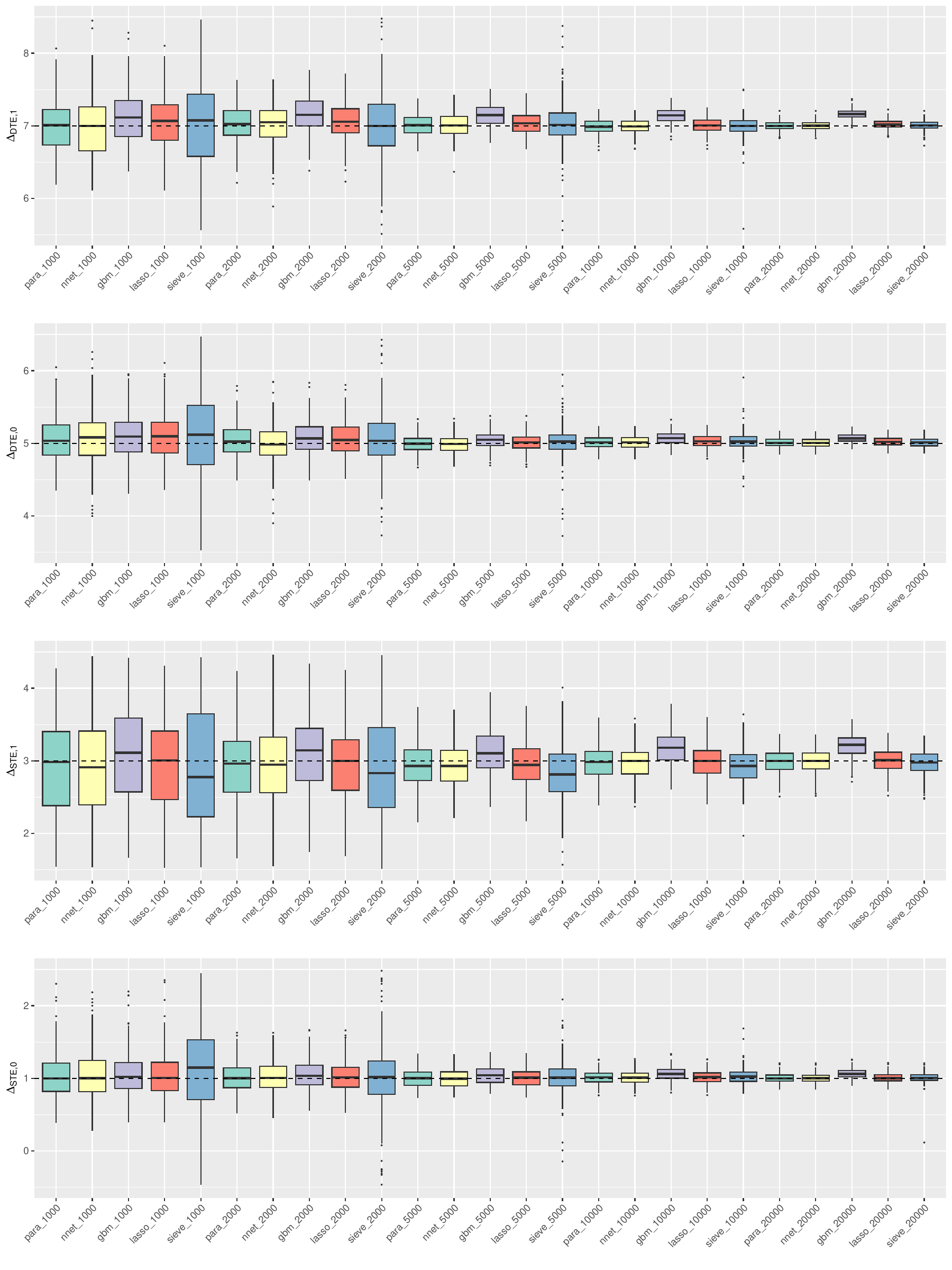}
    \caption{Simulation results of the estimated causal peer effects using various estimation methods.}
    \label{fig:main}
\end{figure}
\section{Application}
\label{sec:app}
In this section, we apply the proposed methods to data from the China Family Panel Studies (CFPS), conducted by the Institute of Social Science Survey, Peking University, between 2010 and 2018. Since its inception in 2010, the CFPS has released five waves of data, covering 25 provinces and representing 95\% of China's population. The data include typical demographic, socio-economic, education, and health information from respondents. We will investigate the impact of retirement on cognitive health and further explore whether retirement may create spillover effects between couples. This study aims to understand the potential effects of retirement on both individuals and their spouses, providing policymakers with insights into retirement systems and health interventions. After excluding individuals who did not report retirement, outliers, missing samples, and removing observations where spouses were not included in the survey, we ended up with paired samples of 1373 couples, totaling 2746 individuals.

In this study, $D_1$ indicates whether males are retired (1 for yes, 0 for no), and $Y_1$ represents males' fluid cognitive health, which reflects the brain's ability to process information, including memory, numeracy, verbal fluency, orientation, logic, and reaction time. Similarly, $D_2$ indicates whether females are retired (1 for yes, 0 for no), and $Y_2$ represents females' fluid cognitive health. Two IVs $Z_1$ and $Z_2$ are selected to indicate whether individuals have reached the retirement age: $Z_1$ represents whether males have reached the retirement age (1 for yes, 0 for no), and $Z_2$ represents whether females have reached the retirement age (1 for yes, 0 for no). The background variables include the number of family members $X_1$, whether males have health insurance $X_2$ (binary variable, 1 for yes, 0 for no), whether females have health insurance $X_3$ (binary variable, 1 for yes, 0 for no), males' education level $X_4$, females' education level $X_5$, whether the family has urban household registration $X_6$ (binary variable, 1 for yes, 0 for no).  The demographic statistics of the observed variables are provided in Table \ref{tab:realdemo-data}.
\begin{table}[t]
\centering
\caption{Demographic statistics of observed variables, including the sample mean, sample standard deviation, and the observed range.}
\label{tab:realdemo-data}
\resizebox{0.824799486795\textwidth}{!}{
\begin{tabular}{@{}ccccc@{}}
\toprule
\multicolumn{1}{c}{{Variable}} & {Description} & {Sample mean} & {Sample SD} & { Sample range} \\ 
\midrule
$X_1$ & Household size              & 4.13  & 1.81  &  $[2, 16] $     \\
$X_2$ & Male insurance status      & 0.02  & 0.12  & 0,1                 \\
$X_3$ & Female insurance status    & 0.02  & 0.15  & 0,1                 \\
$X_4$ & Male education level        & 2.78  & 1.18  &  $[1, 7]  $     \\
$X_5$ & Female education level      & 2.34  & 1.22  &  $[1, 7] $      \\
$X_6$ & Urban household registration & 0.28  & 0.45  & 0,1                 \\
$Z_1$ & Male  reaching retirement age & 0.23  & 0.42  & 0,1                 \\
$Z_2$ & Female reaching retirement age & 0.87  & 0.33  & 0,1                 \\
$D_1$ & Male retirement status      & 0.27  & 0.44  & 0,1                 \\
$D_2$ & Female retirement status    & 0.27  & 0.45  & 0,1                 \\
$Y_1$ & Male fluid cognitive health & 0.10  & 0.68  &  $[-1.95, 1.89] $\\
$Y_2$ & Female fluid cognitive health & -0.10 & 0.81  &  $[-1.97, 2.22]$ \\
\bottomrule
\end{tabular}}
\end{table}
 We then separately examine the direct effects and spillover effects of retirement on fluid cognitive health for both males and females.   Table \ref{tab:real-data}  presents the estimation results based on 200 bootstrap samples. We adopt multiple robust estimation methods similar to those used in the simulation study, including parametric models and several machine learning approaches. Point estimates that significantly exceed the bounds in Table \ref{tab:realdemo-data} are marked with an asterisk (*), indicating that these results may not be reliable. Additionally, bootstrap estimation results that exceed the range are excluded. 
We find that for most results in Table \ref{tab:real-data}, the 95\% confidence intervals (CIs) are not statistically significant. However, for males, the multiply robust estimator based on parametric models suggests that $\widehat{\Delta}_{\STEone} = -1.53$, with a 95\% CI of $(-2.87, -0.19)$, indicating that when males retire, the retirement status of females significantly affects male cognitive health. Additionally, similar results are observed with the multiply robust estimators based on \texttt{gbm} and \texttt{lasso}. This effect may be attributed to reduced spousal interaction or changes in family roles after retirement, which decrease men's engagement in cognitive activities and negatively impact their mental health \citep{xue2018effect,meng2017impact}.

\begin{table}[t]
\centering
\caption{Point estimates of the causal peer effects of retirement age on fluid cognitive health, with 95\% confidence intervals provided in parentheses.
}
\label{tab:real-data}
\resizebox{0.99486795\textwidth}{!}{
\begin{tabular}{ccccc}
  \toprule
Male & $\Delta_\DTEone$ & $\Delta_\DTEzero$ & $\Delta_\STEone$ & $\Delta_\STEzero$ \\  \hline\addlinespace[1mm]
Parametric & -2.12 (-4.89, 0.65) & 15.67$^*$ (11.8, 19.54) &  \bf -1.53 (-2.87, -0.19) & -1.42 (-4.28, 1.45) \\ 
  nnet & -0.22 (-2.94, 2.50) & 104.93$^*$ (100.89, 108.96) & 0.07 (-2.45, 2.58) & 5.31$^*$ (2, 8.62) \\ 
  gbm & -1.94 (-4.78, 0.91) & -41.59$^*$ (-45.32, -37.87) & \bf -1.05 (-2.1, 0) & 5.07$^*$ (2.34, 7.80) \\ 
 lasso & -1.06 (-3.52, 1.41) & 0.02 (-3.88, 3.91) &  \bf -1.09 (-2.23, 0) & -0.11 (-2.81, 2.59) \\ 
  sieve & -63.59$^*$ (-66.71, -60.46) & 1.81 (-2.89, 6.50) & -5.21$^*$ (-6.82, -3.60) & -1.05 (-3.53, 1.43) \\   \toprule Female & $\Delta_\DTEone$ & $\Delta_\DTEzero$ & $\Delta_\STEone$ & $\Delta_\STEzero$ \\ \hline\addlinespace[1mm]
  Parametric & 0.04 (-0.62, 0.7) & -0.93 (-4.18, 2.33) & -2.85 (-5.95, 0.26) & -25.81$^*$ (-30.77, -20.85) \\ 
  nnet & -0.15 (-2.35, 2.05) & -1.46 (-5.49, 2.57) & -4.24$^*$ (-7.27, -1.21) & -0.37 (-5.02, 4.27) \\ 
  gbm & -0.03 (-0.65, 0.58) & -0.99 (-3.72, 1.74) & -18.43$^*$ (-21.1, -15.76) & -3.48 (-7.4, 0.44) \\ 
  lasso & -0.02 (-0.64, 0.60) & -0.28 (-3.04, 2.48) & -22.58$^*$ (-25.76, -19.39) & 4.64$^*$ (-0.14, 9.42) \\ 
  sieve & -36.34$^*$ (-37.27, -35.42) & -1.46 (-4.09, 1.18) & -76.62$^*$ (-80.67, -72.56) & -34.71$^*$ (-39.56, -29.87) \\ 
   \toprule
\end{tabular}}
\end{table}

\section{Extension to average interaction effect}
\label{sec:exten-AIE}
In addition to the causal estimands discussed in the previous sections, the interaction between two treatments and its impact on the overall effect in the presence of interference is also  important. Formally, the average interaction effect is defined as the difference between the two direct effect estimands or the two spillover effect estimands, which can be expressed as:
\begin{align}
\label{eq:ITE-estimand}
    \Delta _{\ite}
   =\Delta _{\DTEone}-\Delta _{\DTEzero} =E\{Y_1(1,1)-Y_1(0,1)\}-E\{Y_1(1,0)-Y_1(0,0)\}.
\end{align} 
This captures the difference in the direct effect experienced by an individual when their partner receives treatment versus not. Since $\Delta _{\ite}$ can also be equivalently expressed as $\Delta _{\ite} = \Delta _{\STEone} - \Delta _{\STEzero}$, Theorem \ref{thm:iden} not only identifies the direct (spillover)   treatment effect but also ensures the identification of the average interaction effect. However, if we only focus on the causal interaction effect $\Delta _{\ite}$, Assumption  \ref{assump:iv}(c) and  Assumption \ref{assump:no-interaction} can be further relaxed. We summarize the corresponding Assumption \ref{assump-ite:iv} below from the direct effect difference perspective.
 \begin{assumption}[Single IV] 
 \label{assump-ite:iv}
 We consider the following assumptions for $d\in \{  0,1 \}$: 
 \begin{itemize} 
     \item[(a)]  $Z_1\indep Y_1\mid (D_1,D_2,U,X,Z_2)$.   
     \item[(b)]  $ 
 \operatorname{cov}\{\Gamma_{\ITE} (Z_2,X, U), \delta_{ \ITE} (Z_2,X, U) \mid Z_2,X\}=0,  $ where  
  \begin{gather*}\Gamma_{\ITE}  \left(Z_2,X,U\right) \equiv E\left\{Y_1\left( 1,d\right)-Y_1\left(0,d\right)\mid Z_2,X,U\right\} \text{ and}\\
     \delta_{\ITE } \left( Z_2,X, U\right)  =  E\left\{D_1 \mathbb{I}(D_2=d) \mid Z_1=1,Z_2 ,X,U\right\} - E\left\{D_1 \mathbb{I}(D_2=d) \mid Z_1=0 ,Z_2,X,U\right\}.  
 \end{gather*}
 \end{itemize}
 \end{assumption} 
Here, we actually have $\delta_{\ITE}(Z_2,X,U)=\delta_{\DTE}(Z_2,X,U)$ because the above conditions are based on the difference of two direct effect estimands. To identify $\Delta_{\ite}$, Assumption \ref{assump-ite:iv} indicates that $Z_2$ does not need to act as an IV and can simply be treated as an additional covariate.  We also define some associated nuisance functions. Let $\pi_{\ite}(Z_1,Z_2,X)=\pr(Z_1\mid Z_2,X)$,  $\eta_{\ITE}(Z_2,X)=E\{Y_1 \mathbb{I}(D_2=d) \mid Z_1=0,Z_2,X\}$,  $\mu_{\ITE}(Z_2,X)= E\{D_1 \mathbb{I}(D_2=d)\mid Z_1=0,Z_2, X\}$, and  $\delta_{\ITE}(Z_2,X)= \pr(D_{1}=1,D_2=d\mid Z_{{1}}=1,Z_2,X)-\pr(D_{1}=1,D_2=d\mid Z_{{1}}=0,Z_2,X) $. Also, we denote $\omega_{\ite}(Z_2,X) =  \omega_{\ITEone}(Z_2,X) - \omega_{\ITEzero}(Z_2,X)$ with $\omega_{\ITE}(Z_2,X) = E\{ Y_1(1,d) - Y_1(0,d) \mid Z_2,X \}$. Motivated by Theorem \ref{thm:iden}, we expect $\omega_{\ITE}(Z_2,X)$ to be as follows
\begin{align}
\label{eq3}
\omega_{\ITE}(Z_2,X) = \frac{ E\{ \mathbb{I}(D_2=d) Y_1 \mid Z_2,X,Z_1=1 \} - E\{ \mathbb{I}(D_2=d) Y_1 \mid Z_2,X,Z_1=0 \} }{ E\{ \mathbb{I}(D_2=d) D_1 \mid Z_2,X,Z_1=1 \} - E\{ \mathbb{I}(D_2=d) D_1 \mid Z_2,X,Z_1=0 \} }.
\end{align} The theoretic verification of \eqref{eq3} is provided in the Supplementary Material.  Alternative identification expressions are also provided in the following proposition.
       \begin{proposition}
    \label{thm:iden-ite} 
    Under Assumptions \ref{assump:iv}(a,b,d), \ref{assump:latent_overlap} and \ref{assump-ite:iv}, 
    We provide three representation formulas as follows, each of which involves a distinct set of nuisance parameters:
    
    (a) Explicit representation:  $0=E\{\varphi_{\ite}^1(O;\pi_{\ite},\delta_{\ITE} \mid Z_2,X\}-\omega_{\ITE}(Z_2,X)$ almost surely, where
    \begin{align*}
\varphi_{\ite}^1(O;\pi_{\ite},\delta_{\ITE}  ) =\frac{(-1)^{1-Z_1} \mathbb{I}(D_2=d) Y_1}{\pi_{\ite}(Z_1,Z_2,X)\delta_{\ITE}(Z_2,X)}.
    \end{align*}

    (b) Implicit representation:  $0=E\{\varphi^2_{\ITE}(O;\pi_{\ite},\omega_{\ITE} )\mid Z_2,X\} $ almost surely, where
    \begin{align*}
        \varphi_{\ITE}^2(O;\pi_{\ite},\omega_{\ITE} )  = & \frac{(-1)^{1-Z_1} \mathbb{I}(D_2=d) }{\pi_{\ite}(Z_1,Z_2,X) } \{ Y_1- D_1 \omega_{\ITE}(Z_2,X) \}.
    \end{align*}

    (c) Implicit representation:  $0=E\{\varphi^3_{\ITE}(O;\mu_{\ITE}, \eta_{\ITE}, \omega_{\ITE})\mid Z_2,Z_1,X\} $ almost surely, where
     \begin{gather*}
            \begin{aligned}
\varphi_{\ITE}^3&(O;\mu_{\ITE}, \eta_{\ITE}, \omega_{\ITE})\\&  = 
  \mathbb{I}( D_2=d)  \left\{  Y_1-D_1\omega_{\ITE}(Z_2,X)\right\} -\eta_{\ITE}(Z_2,X) +\mu_{\ITE}(Z_2,X)\omega_{\ITE}(Z_2,X).\end{aligned}        
     \end{gather*} 
           
       \end{proposition}

Proposition \ref{thm:iden-ite} identifies the average interaction effect $\Delta_{\ite} = \mathbb{E}\{\omega_{\ITEone}(Z_2, X) - \omega_{\ITEzero}(Z_2, X)\}$ under Assumption  \ref{assump:iv}(a-b) {{with $Z_2$ being in the conditional argument}}, the first condition of Assumption~\ref{assump:iv}(d), {{Assumption~\ref{assump:latent_overlap}}} and Assumption \ref{assump-ite:iv}. It provides three different representations, offering flexibility in identifying $\Delta_{{\ite}}$ by leveraging different sets of nuisance parameters. Moreover, the various expressions in this section naturally inspire the development of multiply robust estimators and their corresponding theoretical results. We omit the discussion on estimation methods and related conclusions,  since all results can be established in a parallel manner as in Section \ref{sec:estimation}.
\section{Conclusion}
\label{sec:discussion}
IVs are widely used to identify causal effects in the presence of unmeasured confounding \citep{Angrist:1996,Wang2017iv}. In many practical applications, however, the focus of IV literature predominantly lies within frameworks that assume no spillover effects. To address these issues, we employ dual IVs under a series of assumptions to {{define, identify and estimate meaningful causal quantities in the presence of spillover effects.}} For more efficient estimation, we also focus on deriving EIFs and demonstrate how they can be utilized to construct estimators based on statistical and machine-learning methods. It is important to highlight that, in practice, correctly specifying several nuisance models may sometimes be challenging. {{So we provide a non-parametric estimation approach as a complement, based on incrementally imposing moment restrictions \citep{CHEN20075549,ai2022simple}.}} These estimators do not require modeling the functional relationships while maintaining consistency and efficiency. Simulation studies show that the proposed estimators perform well in finite samples.  

The proposed methods can be improved or extended in several ways. First,  while this paper primarily focuses on binary instruments and treatments, extending the proposed methodology to accommodate more general instruments or treatments represents an intriguing future research direction. {{Moreover, estimating network effects that extend beyond dyadic data scenarios, especially in the presence of unmeasured network confounding, remains a significant research area. These considerations extend beyond the scope of the current study and present valuable opportunities for future research.}}


 \bibliographystyle{apalike}
					\bibliography{mybib}
 
 \newpage
  
\hypersetup{linkcolor=black}
\renewcommand{\thesection}{S\arabic{section}}
\renewcommand{\theequation}{S\arabic{equation}}  
\renewcommand{\thefigure}{S\arabic{figure}} 
\renewcommand{\thetable}{S\arabic{table}} 
\renewcommand{\thetheorem}{S\arabic{theorem}} 
\renewcommand{\thelemma}{S\arabic{lemma}} 
\renewcommand{\theproposition}{S\arabic{proposition}} 
\renewcommand{\theassumption}{S\arabic{assumption}}  

\hypersetup{linkcolor=blue}
\begin{center}
  \bf \LARGE  Supplementary  Material
\end{center}
\vspace{5mm}
\appendix
The Supplementary Material includes proofs of the theorems and propositions, as well as the estimation procedure.


\section{{{Some notations for spillover effects}}}
\label{notation_spillover}
\begin{table}[ht]
	\centering
	\caption{Summary of notations for spillover effects used in this paper.}
	\label{tab:notations_spillover}
	\begin{tabular}{ccc}
		\toprule
    Notation & Definition  & Description   \\
        \midrule
		 $\pi_{\ste}(Z_2,X)$ & $\pr(Z_2\mid X)$    & Instrumental propensity score   \\
		$\delta_{\DTE}^j(X)$	& $\pr(D_{1}=d,D_2=1 \mid Z_2=j,X)$    &  Conditional treatment probability  \\
		 $\delta_{\STE}(X)$  & $\delta_{\STE}^1(X) - \delta_{\STE}^0(X)$   &  Treatment probability difference  \\
         $\mu_{\STE}(X)$ & $E\{ \mathbb{I}(D_1=d) D_2 \mid Z_2=0 ,X\}$   &  Conditional treatment expectation  \\
		 $\eta_{\STE}(X)$ & $E\{ Y_1\mathbb{I}(D_1=d)\mid Z_2=0,X\}$  &  Conditional outcome expectation \\
       {{$\omega_{\STE}(X)$}}  &  $E\{  Y_1(d,1) - Y_1(d,0) \mid X \}$  & Conditional effect expectation\\
		\bottomrule
	\end{tabular}%
\end{table}%

\section{The proof of Theorem \ref{thm:iden}}
\begin{proof} 
Assumption \ref{assump:iv}(a) implies that $(Z_1, Z_2) \indep U \mid X$, which guarantees that $(Z_1, U) \indep Z_2 \mid X$ and $(Z_2, U) \indep Z_1 \mid X$. Here, we fix $d=1$ just for convenience. We first consider the proof of average direct effect,
\begin{align*}
&E\left(Y_1D_2\mid Z_1=1 ,X \right) -E\left(Y_1D_2\mid Z_1=0 ,X \right) \\&=E\left\{E\left(Y_1D_2\mid Z_1=1,Z_2,U,X \right) \mid Z_1=1 ,X \right\} -E\left\{E\left(Y_1D_2\mid Z_1=0,Z_2,U,X \right) \mid Z_1=0 ,X \right\} \\&=E\left\{E\left(Y_1D_2\mid Z_1=1,Z_2,U,X \right) \mid X \right\}-E\left\{E\left(Y_1D_2\mid Z_1=0,Z_2,U,X \right) \mid X \right\}\;\;\;  {~~ Z_1 \indep (Z_2,U)\mid X }\\
&=E\left\{E\left(Y_1D_1D_2\mid Z_1=1,Z_2,U,X \right) \mid X \right\}  +E\left[E\left\{Y_1\left(1-D_1 \right) D_2\mid Z_1=1, Z_2,U,X \right\} \mid X\right]\\&\;\;\;\;\;\;\;-E\left\{E\left(Y_1D_1D_2\mid Z_1=0,Z_2,U,X \right) \mid X \right\} -E\left[E\left\{Y_1\left(1-D_1 \right) D_2\mid  Z_1=0,Z_2,U,X \right\}\mid X\right]\\
&=E\left[E\left\{Y_1(1,1)D_1D_2\mid Z_1=1,Z_2,U,X  \right\} \mid X \right] +E\left[E\left\{Y_1(0,1)\left(1-D_1 \right) D_2\mid Z_1=1,Z_2,U,X \right\} \mid X\right]\\&\;\;\;\;\;\;\;-E\left[E\left\{Y_1(1,1)D_1D_2\mid Z_1=0,Z_2,U,X  \right\} \mid X \right] -E\left[E\left\{Y_1(0,1)\left(1-D_1 \right) D_2\mid Z_1=0,Z_2,U,X \right\} \mid X\right]\\
&=E\left[E\left\{Y_1(1,1)\mid  U,X \right\} \mathrm{pr}(D_1=1,D_2=1\mid Z_1=1,Z_2,U,X) \mid X\right] { \qquad \text{Assumption}\;\ref{assump:latent_overlap}(a)}\\&\;\;\;\;\;\;\;\;+E\left[E\left\{Y_1(0,1)\mid U,X \right\} \mathrm{pr}(D_1=0,D_2=1\mid Z_1=1,Z_2,U,X)\mid X\right]\\&\;\;\;\;\;\;\;-E\left[E\left\{Y_1(1,1)\mid U,X \right\} \mathrm{pr}(D_1=1,D_2=1\mid Z_1=0,Z_2,U,X)\mid X\right]\\&\;\;\;\;\;\;\;\;-E\left[E\left\{Y_1(0,1)\mid  U,X \right\} \mathrm{pr}(D_1=0,D_2=1\mid Z_1=0,Z_2,U,X)\mid X\right]\;\\&=E\left[E\left\{Y_1(1,1)\mid  U,X \right\} \mathrm{pr}(D_2=1\mid Z_1=1,Z_2,D_1=1,U,X) \mathrm{pr}(D_1=1\mid Z_1=1,Z_2,U,X)\mid X\right]\\&\;\;\;\;\;\;\;\;+E\left[E\left\{Y_1(0,1)\mid U,X \right\} \mathrm{pr}(D_2=1\mid Z_1=1,Z_2,D_1=0,U,X) \mathrm{pr}(D_1=0\mid Z_1=1,Z_2,U,X)\mid X\right]\\&\;\;\;\;\;\;\;-E\left[E\left\{Y_1(1,1)\mid U,X \right\} \mathrm{pr}(D_2=1\mid Z_1=0,Z_2,D_1=1,U,X) \mathrm{pr}(D_1=1\mid Z_1=0,Z_2,U,X)\mid X\right]\\&\;\;\;\;\;\;\;\;-E\left[E\left\{Y_1(0,1)\mid  U,X \right\} \mathrm{pr}(D_2=1\mid Z_1=0,Z_2,D_1=0,U,X) \mathrm{pr}(D_1=0\mid Z_1=0,Z_2,U,X)\mid X\right]\\& =E\left[E\left\{Y_1(1,1)\mid  U,X \right\} \mathrm{pr}(D_2=1\mid Z_2,U,X) \mathrm{pr}(D_1=1\mid Z_1=1,Z_2,U,X)\mid X\right]\\&\;\;\;\;\;\;\;\;+E\left[E\left\{Y_1(0,1)\mid U,X \right\}  \mathrm{pr}(D_2=1\mid  Z_2,U,X) \mathrm{pr}(D_1=0\mid Z_1=1,Z_2,U,X)\mid X\right]\\&\;\;\;\;\;\;\;-E\left[E\left\{Y_1(1,1)\mid U,X \right\} \mathrm{pr}(D_2=1\mid  Z_2,U,X) \mathrm{pr}(D_1=1\mid Z_1=0,Z_2,U,X)\mid X\right]\\&\;\;\;\;\;\;\;\;-E\left[E\left\{Y_1(0,1)\mid  U,X \right\} \mathrm{pr}(D_2=1\mid  Z_2,U,X) \mathrm{pr}(D_1=0\mid Z_1=0,Z_2,U,X)\mid X\right] 
\\&=E\left[E\left\{Y_1(1,1)\mid  U,X \right\}\right.\\&\;\;\;\;\;\;\;\;\;\;\;\left.\times\left\{\mathrm{pr}(D_1=1\mid Z_1=1,Z_2,U,X)-\mathrm{pr}(D_1=1\mid Z_1=0,Z_2,U,X)\mid X \right\} \mathrm{pr}(D_2=1\mid Z_2,U,X) \right]\\&\;\;\;\;\;\;\;\;-E\left[E\left\{Y_1(0,1)\mid  U,X \right\}\right.\\&\;\;\;\;\;\;\;\;\;\;\;\left.\times\left\{\mathrm{pr}(D_1=1\mid Z_1=1,Z_2,U,X)-\mathrm{pr}(D_1=1\mid Z_1=0,Z_2,U,X)\mid X \right\} \mathrm{pr}(D_2=1\mid  Z_2,U,X) \right] 
\\&=E\left(\left[E\left\{Y_1(1,1)\mid  U,X \right\}-E\left\{Y_1(0,1)\mid  U,X \right\}\right]\right.\\&\;\;\;\;\;\;\;\;\;\;\;\left.\times\left\{\mathrm{pr}(D_1=1,D_2=1\mid Z_1=1,Z_2,U,X)-\mathrm{pr}(D_1=1,D_2=1\mid Z_1=0,Z_2,U,X)\right\} \mid X\right)
\\&=E\left[E\left\{Y_1(1,1)-Y_1(0,1)\mid  U,X \right\}\right.\\&\;\;\;\;\;\;\;\;\;\;\;\left.\times\left\{\mathrm{pr}(D_1=1,D_2=1\mid Z_1=1,Z_2,U,X)-\mathrm{pr}(D_1=1,D_2=1\mid Z_1=0,Z_2,U,X)\right\} \mid X\right)
\\
&  =E\left[E\left\{Y_1(1,1)-Y_1(0,1)\mid  U,X \right\}\mid X\right]\\&\;\;\;\;\;\;\;\;\;\;\;  \times E\left[\left\{E(D_1D_2\mid Z_1=1,Z_2,U,X)-E(D_1D_2\mid Z_1=0,Z_2,U,X)\mid X \right\}\mid X\right] \;\;{ \text{Assumption}\;\ref{assump:no-interaction}(a)}
\\&=E\left\{Y_1(1,1)-Y_1(0,1)\mid X \right\} \left\{E(D_1 D_2 \mid Z_1=1 ,X)-E(D_1 D_2 \mid Z_1=0 ,X)\mid X \right\}.
\end{align*} 
\end{proof} 
We hence obtain the following expression: \begin{align} 
\label{eq:wald-estimand}
  \omega_{\DTEone} (X)  
=E\left\{Y_1(1,1)-Y_1(0,1)\mid X \right\} = \dfrac{E(Y_1D_2 \mid Z_1=1,  X)-E\left(Y_1 D_2\mid Z_1=0 ,X \right)}{ E(D_1 D_2\mid Z_1=1, X )-E(D_1D_2  \mid Z_1=0, X )} . 
\end{align} 

Similarly, for $\omega_{\STE}(X)$, we have
\begin{proof} The proof of spillover treatment effect:
\begin{align*}
 &E\left(Y_1D_1\mid Z_2=1,X \right)-E\left(Y_1D_1\mid Z_2=0 ,X\right)
 \\&=E\left\{E\left(Y_1D_1\mid Z_2=1,Z_1,U,X\right)\mid Z_2=1,X \right\} -E\left\{E\left(Y_1D_1\mid Z_2=0,Z_1,U,X\right)\mid Z_2=0,X \right\}
 \\&=E\left\{E\left(Y_1D_1\mid Z_2=1,Z_1,U,X\right) \mid X\right\}-E\left\{E\left(Y_1D_1\mid Z_2=0,Z_1,U,X\right) \mid X\right\}\;\;\;  {~~ Z_2\indep (Z_1,U) \mid X }
 \\&=E\left\{E\left(Y_1D_2D_1\mid Z_2=1,Z_1,U,X\right)\mid X\right\}  +E\left[E\left\{Y_1\left(1-D_2\right)D_1\mid Z_2=1, Z_1,U,X\right\} \mid X \right]
 \\&\;\;\;\;\;\;\;-E\left\{E\left(Y_1D_2D_1\mid Z_2=0,Z_1,U,X\right) \mid X\right\} -E\left[E\left\{Y_1\left(1-D_2\right)D_1\mid  Z_2=0,Z_1,U,X\right\} \mid X\right]
 \\&=E\left\{E\left(Y_1(1,1)D_2D_1\mid Z_2=1,Z_1,U,X\right)\mid X \right\}\\&\;\;\;\;\;\;\;+E\left[E\left\{Y_1(1,0)\left(1-D_2\right)D_1\mid Z_2=1,Z_1,U,X\right\}\mid X \right]\\&\;\;\;\;\;\;\;-E\left\{E\left(Y_1(1,1)D_2D_1\mid Z_2=0,Z_1,U,X\right)\mid X \right\}\\&\;\;\;\;\;\;\;-E\left[E\left\{Y_1(1,0)\left(1-D_2\right)D_1\mid Z_2=0,Z_1,U,X\right\} \mid X \right] \\&=E\left[E\left\{Y_1(1,1)\mid  U\right\} \mathrm{pr}(D_2=1,D_1=1\mid Z_2=1,Z_1,U,X) \mid X \right] { \qquad \text{Assumption}\;\ref{assump:latent_overlap}(a)}\\&\;\;\;\;\;\;\;\;+E\left[E\left\{Y_1(1,0)\mid U,X\right\} \mathrm{pr}(D_2=0,D_1=1\mid Z_2=1,Z_1,U,X) \mid X\right]\\&\;\;\;\;\;\;\;-E\left[E\left\{Y_1(1,1)\mid U\right\} \mathrm{pr}(D_2=1,D_1=1\mid Z_2=0,Z_1,U,X) \mid X\right]\\&\;\;\;\;\;\;\;\;-E\left[E\left\{Y_1(1,0)\mid  U,X\right\} \mathrm{pr}(D_2=0,D_1=1\mid Z_2=0,Z_1,U,X)\right]\;\\&=E\left[E\left\{Y_1(1,1)\mid  U,X\right\} \mathrm{pr}(D_1=1\mid Z_2=1,Z_1,D_2=1,U,X) \mathrm{pr}(D_2=1\mid Z_2=1,Z_1,U,X) \mid X\right]\\&\;\;\;\;\;\;\;\;+E\left[E\left\{Y_1(1,0)\mid U,X\right\} \mathrm{pr}(D_1=1\mid Z_2=1,Z_1,D_2=0,U,X) \mathrm{pr}(D_2=0\mid Z_2=1,Z_1,U,X) \mid X\right]\\&\;\;\;\;\;\;\;-E\left[E\left\{Y_1(1,1)\mid U,X\right\} \mathrm{pr}(D_1=1\mid Z_2=0,Z_1,D_2=1,U,X) \mathrm{pr}(D_2=1\mid Z_2=0,Z_1,U,X) \mid X\right]\\&\;\;\;\;\;\;\;\;-E\left[E\left\{Y_1(1,0)\mid  U,X\right\} \mathrm{pr}(D_1=1\mid Z_2=0,Z_1,D_2=0,U,X) \mathrm{pr}(D_2=0\mid Z_2=0,Z_1,U,X) \mid X\right]\\& =E\left[E\left\{Y_1(1,1)\mid  U,X\right\} \mathrm{pr}(D_1=1\mid Z_1,U,X) \mathrm{pr}(D_2=1\mid Z_2=1,Z_1,U,X) \mid X\right]\\&\;\;\;\;\;\;\;\;+E\left[E\left\{Y_1(1,0)\mid U,X\right\} \times\mathrm{pr}(D_1=1\mid  Z_1,U,X) \mathrm{pr}(D_2=0\mid Z_2=1,Z_1,U,X)\mid X\right]\\&\;\;\;\;\;\;\;-E\left[E\left\{Y_1(1,1)\mid U,X\right\} \mathrm{pr}(D_1=1\mid  Z_1,U,X) \mathrm{pr}(D_2=1\mid Z_2=0,Z_1,U,X) \mid X\right]\\&\;\;\;\;\;\;\;\;-E\left[E\left\{Y_1(1,0)\mid  U,X\right\} \mathrm{pr}(D_1=1\mid  Z_1,U,X) \mathrm{pr}(D_2=0\mid Z_2=0,Z_1,U,X) \mid X\right] 
 \\&=E\left[E\left\{Y_1(1,1)\mid  U,X\right\}\right.\\&\;\;\;\;\;\;\;\;\;\;\;\left.\times\left\{\mathrm{pr}(D_2=1\mid Z_2=1,Z_1,U,X)-\mathrm{pr}(D_2=1\mid Z_2=0,Z_1,U,X) \mid X\right\} \mathrm{pr}(D_1=1\mid Z_1,U,X)\mid X \right]\\&\;\;\;\;\;\;\;\;-E\left[E\left\{Y_1(1,0)\mid  U,X\right\}\right.\\&\;\;\;\;\;\;\;\;\;\;\;\left.\times\left\{\mathrm{pr}(D_2=1\mid Z_2=1,Z_1,U,X)-\mathrm{pr}(D_2=1\mid Z_2=0,Z_1,U,X) \mid X\right\} \mathrm{pr}(D_1=1\mid  Z_1,U,X) \mid X \right] 
 \\&=E\left[E\left\{Y_1(1,1)\mid  U,X\right\}-E\left\{Y_1(1,0)\mid  U,X\right\}\right.\\&\;\;\;\;\;\;\;\;\;\;\;\left.\times\left\{\mathrm{pr}(D_2=1,D_1=1\mid Z_2=1,Z_1,U,X)-\mathrm{pr}(D_2=1,D_1=1\mid Z_2=0,Z_1,U,X)\right\} \mid X \right]
 \\&=E\left[E\left\{Y_1(1,1)-Y_1(1,0)\mid  U,X\right\}\right.\\&\;\;\;\;\;\;\;\;\;\;\;\left.\times\left\{\mathrm{pr}(D_2=1,D_1=1\mid Z_2=1,Z_1,U,X)-\mathrm{pr}(D_2=1,D_1=1\mid Z_2=0,Z_1,U,X)\right\} \mid X \right]
 \\
 & =E\left[E\left\{Y_1(1,1)\mid  U,X\right\}-E\left\{Y_1(1,0)\mid  U,X\right\} \mid X\right]\\&\;\;\;\;\;\;\;\;\;\;\; \times E\left[\left\{E(D_2D_1\mid Z_2=1,Z_1,U,X)-E(D_2D_1\mid Z_2=0,Z_1,U,X)\right\} \mid X\right] \;\;  \text{Assumption}\;\ref{assump:no-interaction}(b)
  \\&=E\left\{Y_1(1,1)-Y_1(1,0)\mid X\right\} \left\{E(D_2 D_1 \mid Z_2=1,X )-E(D_2 D_1 \mid Z_2=0,X )\right\}.
\end{align*} 
\end{proof}
We hence obtain the following expression: \begin{align} 
  \omega_{\STEone} (X)  
=E\left\{Y_1(1,1)-Y_1(1,0)\mid X \right\} = \dfrac{E(Y_1D_1 \mid Z_2=1,  X)-E\left(Y_1 D_1\mid Z_2=0 ,X \right)}{ E(D_2 D_1\mid Z_2=1, X )-E(D_2D_1  \mid Z_2=0, X )} . 
\end{align} 
\section{The proof of Corollary \ref{cor:alternative-form}}
\begin{proof}
    We now verify the results in Corollary \ref{cor:alternative-form}(a), 
    \begin{align*}
 E\left\{\frac{(-1)^{1-Z_1}D_2Y_1}{\pi_{\dte}(Z_1,X)\delta_{\DTEone}(X)}\bigg|X\right\}= \dfrac{E(Y_1D_2 \mid Z_1=1,  X)-E\left(Y_1 D_2\mid Z_1=0 ,X \right)}{\delta_{\DTEone}(X)} = \omega_{\DTE}(X). 
    \end{align*}
    Next, we verify the results in Corollary \ref{cor:alternative-form}(b), 
    \begin{align*}
 E & \left\{ \frac{(-1)^{1-Z_1} D_2}{\pi_{\dte}(Z_1,X) }\left\{Y_1-D_1 
\omega_{\DTEone}(X)\right\}\bigg|X\right\} \\&~~~~ = \{ {E(Y_1D_2 \mid Z_1=1,  X)-E\left(Y_1 D_2\mid Z_1=0 ,X \right)}\} -\delta_{\DTEone}(X) \Delta_{{\DTEone}} (X), 
    \end{align*}
    \begin{align*}
       \Rightarrow \omega_{\DTEone}(X)=\dfrac{E(Y_1D_2 \mid Z_1=1,  X)-E\left(Y_1 D_2\mid Z_1=0 ,X \right)}{\delta_{\DTEone}(X)}.
    \end{align*}   Finally, we directly verify the moment restriction in Corollary \ref{cor:alternative-form}(c), 
    \begin{align*}
 E & \left[\left\{ D_2 Y_1-\eta_{\DTEone}(X)-D_1 D_2 \omega_{\DTEone}(X)  +\mu_{\DTEone}(X)\omega_{\DTEone}(X)\right\}\bigg|Z_1,X\right]\\&~~~~ = E\left(Y_1 D_2\mid Z_1 ,X \right)  -\eta_{\DTEone}(X)- E\left(D_1 D_2\mid Z_1 ,X \right)\omega_{\DTEone}(X)    +\mu_{\DTEone}(X)\omega_{\DTEone}(X)    \\&~~~~ =   \eta_{\DTEone}(X)+Z_1\omega_{\DTEone}(X) \delta_{\DTEone}(X) -\eta_{\DTEone}(X)- \{ \mu_{\DTEone}(X)+Z_1 \delta_{{\DTEone}}(X)\}\omega_{\DTEone}(X)  \\&~~~~~~  +\mu_{\DTEone}(X)\omega_{\DTEone}(X)\\&~~~~ = 0.
    \end{align*}
    where we use the fact
    \begin{align*}
    E(Y_1D_2\mid Z_1,X)&=    E(Y_1D_2\mid Z_1=0,X)+Z_1\{  E(Y_1D_2\mid Z_1=1,X)-  E(Y_1D_2\mid Z_1=0,X)\}\\&=    E(Y_1D_2\mid Z_1=0,X)+Z_1\omega_{\DTEone}(X) \delta_{{\DTEone}}(X) \\&=    \eta_{\DTEone}(X)+Z_1\omega_{\DTEone}(X) \delta_{{\DTEone}}(X) ,\\ E(D_1D_2\mid Z_1,X)&=    E(D_1D_2\mid Z_1=0,X)+Z_1\{  E(D_1D_2\mid Z_1=1,X)-  E(D_1D_2\mid Z_1=0,X)\}\\ &=  \mu_{\DTEone}(X)+Z_1 \delta_{{\DTEone}}(X) .
\end{align*}
\end{proof}
\section{The proof of Theorem \ref{thm:eif-2}}  
\subsection{Preliminaries}
We will use the semiparametric theory in \citet{bickel1993efficient} to derive the efficient influence functions (EIFs). Let $O=({X},Z_1,Z_2,D_1,D_2,Y_1)$ be the vector of all observed variables. Under Assumption  \ref{assump:iv}, the observed likelihood can be factorized as 
\begin{equation}
\label{eq:likelihood}
f(O)=f(X)f(Z_1\mid X)f( Z_2\mid X)f(Y_1,D_1,D_2\mid Z_1,Z_2,X).
\end{equation}
To derive the EIFs, we consider a one-dimensional parametric submodel, $ f_t(O) $, which includes the true model $ f(O) $ at $ t = 0 $, i.e., $ f_t(O) \big|_{t=0} = f(O) $. We use $ t $ as a subscript to indicate quantities related to the submodel, such as $ \Delta_{\DTEone,t} $, which represents the value of $ \Delta_{\DTEone} $ in the submodel. A dot is used to represent the partial derivative with respect to $ t $, for instance, $ \dot{\Delta}_{\DTEone,t} = \partial \Delta_{\DTEone,t}/\partial t $. The score function of the parametric submodel is denoted by $ s_t(\cdot) $.  The score function $ s_t(O) $ can be decomposed as:
  \[
    s_t(O) = s_t(X) + s_t(Z_1\mid X) + s_t(Z_2\mid X) + s_t( Y_1,D_1,D_2 \mid Z_1,Z_2,X),
    \]
where
  $$ s_t(X) = \frac{\partial \log f_t(X)}{\partial t} , ~~~  s_t(Z_1\mid X) = \frac{\partial \log f_t(Z_1\mid X)}{\partial t},~~~ s_t(Z_2\mid X) = \frac{\partial \log f_t( Z_2\mid X)}{\partial t} ,$$ $$  s_t(Y_1,D_1,D_2 \mid Z_1,Z_2,X) = \frac{\partial \log f_t(Y_1,D_1,D_2 \mid Z_1,Z_2,X)}{\partial t},$$
  are the score functions corresponding to the four components of the likelihood. 
Analogous to $ f_t(O) \big|_{t=0} = f(O) $, we denote $ s_t(\cdot) \big|_{t=0} $ as $ s(\cdot) $, which is the score function evaluated at the true parameter value under the one-dimensional submodel.

From semiparametric theory, the tangent space $ \Lambda = H_1 \oplus H_2 \oplus H_3 \oplus H_4 $ is the direct sum of the following spaces:
  $$ H_1 = \{ h(X) : E\{h(X)\} = 0 \}, $$
  $$ H_2 = \{ h(Z_1, X) : E\{h(Z_1,X) \mid X\} = 0 \}, $$
  $$ H_3 = \{ h(Z_2,X) : E\{h(Z_2,X) \mid X\} = 0 \} , $$
  $$ H_4 = \{ h( Y_1, D_1,D_2,Z_1,Z_2,X) : E\{h( Y_1, D_1,D_2,Z_1,Z_2,X) \mid Z_1,Z_2, X\} = 0 \}, $$ 
  where $ H_1 $, $ H_2 $, $ H_3 $, and $ H_4 $ are orthogonal to each other.  The  EIF for $ \Delta_{\DTE} $, denoted by $ \varphi (O) \in \Lambda $, must satisfy the condition:
  \[
    \dot{\mu}_{\DTE,t} \big|_{t=0} = E\{ \varphi (O)  s(O)\}.
    \]
We will derive the EIFs by calculating $ \dot{\mu}_{\DTE,t} \big|_{t=0} $.
Define the following quantity for any function $ f(Y_1,D_1,D_2, Z_2,X) $:
  \begin{equation*}
\begin{aligned}
\psi_{f(Y_1,D_1,D_2,Z_2, X)} ^{z_1}&= \frac{1(Z_1 = z_1)\left[ f(Y_1,D_1,D_2, Z_2,X) - E\{f(Y_1,D_1,D_2,Z_2, X) \mid X, Z_1 = z_1\} \right]}{\pr(Z_1 = z_1 \mid X)}\\ &~~~~+ E\{f(Y_1,D_1,D_2,Z_2, X) \mid X, Z_1 = z_1\}. 
\end{aligned}
\end{equation*}  
for $z_1=0,1.$ To simplify the proof, we introduce a  lemma.
\begin{lemma}
\label{lem:eif-D1D2Y1}
For any function $ f(Y_1,D_1,D_2,Z_2, X) $, we define $ \mu_{0f}(X) = E\{f(Y_1,D_1,D_2,Z_2, X) \mid Z_1 = 0, X\} $ and $ \mu_{1f}(X) = E\{f(Y_1,D_1,D_2,Z_2, X) \mid Z_1 = 1, X\} $ . The corresponding derivatives of these expectations with respect to $ t $ at $ t = 0 $ are:
  
  \[
    \dot{\mu}_{0f,t}(X) \big|_{t=0} = E\left[\{\psi^0_{f(Y_1, D_1,D_2,Z_2, X)} - \mu_{0f}(X)\} s(Y_1, D_1,D_2, Z_2 \mid Z_1, X) \mid X\right],
    \]
\[
  \dot{\mu}_{1f,t}(X) \big|_{t=0} = E\left[\{\psi^1_{f(Y_1,D_1,D_2, Z_2,X)    } - \mu_{1f}(X)\} s(Y_1, D_1,D_2, Z_2 \mid Z_1, X) \mid X\right].
  \]
As a special case, when $ f(Y_1,D_1,D_2, Z_2,X) = D_1D_2 $, we have $\delta_{\DTEone}^1(X)=E(D_1D_2\mid Z_{{1}}=1,X) $ and $\delta_{\DTEone}^0(X) = E(D_1D_2 \mid Z _1 =0, X) $, where
\[
  \dot{\delta}_{\DTEone}^1(X) \big|_{t=0} = E\left[\{\psi_{D_1D_2}^1 -\delta_{\DTEone}^1(X)\} s(D_1,D_2 \mid Z_1, X) \mid X\right],
  \]
\[
  \dot{\delta}_{\DTEone}^0(X) \big|_{t=0} = E\left[\{\psi_{D_1D_2}^0 - \delta_{\DTEone}^0(X)\} s(D_1,D_2 \mid Z_1, X) \mid X\right].
  \]
As a special case, when $ f(Y_1,D_1,D_2, Z_2,X) = Y_1 D_2 $, we have $\alpha_{\DTEone}^1(X)=E( Y_1 D_2 \mid Z_{{1}}=1,X) $ and $\alpha_{\DTEone}^0(X) = E(Y_1D_2 \mid Z _1 =0, X) $, where
\[
  \dot{\beta}_{\DTEone}^1(X) \big|_{t=0} = E\left[\{\psi_{Y_1D_2}^1 -\alpha_{\DTEone}^1(X)\} s(Y_1,D_2 \mid Z_1, X) \mid X\right],
  \]
\[
  \dot{\beta}_{\DTEone}^0(X) \big|_{t=0} = E\left[\{\psi_{Y_1D_2}^0 - \alpha_{\DTEone}^0(X)\} s(Y_1,D_2 \mid Z_1, X) \mid X\right].
  \]

\end{lemma}
\begin{proof}
We first prove the general result for $ \dot{\mu}_{0f,t}(X) \big|_{t=0} $:
  \begin{align*}
\dot{\mu}_{0f,t}(X) \big|_{t=0} &= \frac{\partial}{\partial t} E_t\{f(Y_1,D_1,D_2, Z_2,X) \mid Z_1 = 0, X\} \Bigg|_{t=0} 
\\
&= E\{f(Y_1,D_1,D_2, Z_2,X) \cdot s(Y_1, D_1,D_2,Z_2 \mid Z_1 = 0, X) \mid Z_1 = 0, X\} \\
&= E\left[\{f(Y_1,D_1,D_2, Z_2,X) - \mu_{0f}(X)\} s(Y_1, D_1,D_2,Z_2 \mid Z_1 = 0, X) \mid Z_1 = 0, X\} 
\right]\\
&= E\left[\{ f(Y_1,D_1,D_2, Z_2,X) - \mu_{0f}(X) \} \cdot \frac{1-Z_1}{\pr(Z_1 =0 \mid X)} \cdot s(Y_1, D_1,D_2,Z_2 \mid Z_1 , X) \mid X\right] \\
&= E\left[\{\psi_{ f(Y_1,D_1,D_2, Z_2,X) }^0- \mu_{0f}(X)\} s(Y_1, D_1,D_2,Z_2 \mid Z_1 , X) \mid X\right],
\end{align*}
where the third equality holds due to the score function satisfying $E\{s(Y_1, D_1,D_2,Z_2 \mid Z_1 = 0, X) \mid Z_1 =  0, X\} = 0 $. Similar reasoning process also applies to $\mu_{1f}(X)$ just by letting $ f(Y_1,D_1,D_2, Z_2,X) =D_1D_2 $ and $ f(Y_1,D_1,D_2, Z_2,X) =Y_1D_2 $:
\[
E\left[\{\psi_{D_1D_2}^{z_1}-\delta_{\DTEone}^{z_1}(X)\} s(D_1,D_2 \mid Z_1, X) \mid X\right] = 0,
\]
\[
E\left[\{\psi_{Y_1D_2}^{z_1}-\alpha_{\DTEone}^{z_1}(X)\} s(Y_1,D_2 \mid Z_1,  X) \mid X\right] = 0.
\]
\end{proof}


 
 
\subsection{EIF for the average direct treatment effect}
 In this section, we show the EIF for   $\Delta_{\DTEone}$. The EIFs for other parameters can be derived similarly. Let  $$\alpha_{\DTE}^1(X)=E\{ Y_1\mathbb{I}(D_{2}=d)\mid Z_1=1,X\} ,~~\alpha_{\DTE}^0(X)= E\{ Y_1\mathbb{I}(D_{2}=d)\mid Z_1=0,X\}.$$ $$\delta_{\DTE}^1(X)=\pr(D_{1}=1,D_2=d\mid Z_{{1}}=1,X) ,~~\delta_{\DTE}^0(X)=\pr(D_{1}=1,D_2=d\mid Z_{{1}}=0,X).$$$$\alpha_{\DTE}(X)=\alpha_{\DTE}^1(X)-\alpha_{\DTE}^0(X),~~\delta_{\DTE}(X)=\delta_{\DTE}^1(X)-\delta_{\DTE}^0(X).$$ First, we have \begin{equation}
     \label{eq:der-delta}
     \begin{aligned} 
\left.  \dot{\Delta}_{\DTEone,t}\right|_{t=0} &=\left.\frac{\partial}{\partial t} E_{t}\left\{\frac{\alpha_{{\DTEone},t}(X)}{\delta_{\DTEone,t}(X)}\right\}\right|_{t=0} \\
&=E\left\{\frac{\alpha_{{\DTEone}}(X)}{\delta_{{\DTEone}}(X)} s(X)\right\}+\left.E\left\{\frac{\dot{\beta}_{{\DTEone},t}(X) \delta_{{\DTEone}}(X)-\alpha_{{\DTEone}}(X) \dot{\delta}_{\DTEone,t}(X)}{\delta^2_{{\DTEone}}(X)}\right\}\right|_{t=0} \\
&=E\left\{\frac{\alpha_{{\DTEone}}(X)}{\delta_{{\DTEone}}(X)}s(X )\right\}  +E\left\{\frac{\dot{\beta}^1_{{\DTEone},t}(X)-\dot{\beta}^0_{{\DTEone},t}(X) }{ \delta_{{\DTEone}}(X) } \right\}\bigg|_{t=0} \\&~~~~~ -E\left[\frac{  \alpha_{{\DTEone}}(X) }{\delta^2_{{\DTEone}}(X)}\{\dot{\delta}^1_{\DTEone,t}(X)-\dot{\delta}^0_{\DTEone,t}(X)\}\right]\bigg|_{t=0} .
\end{aligned}
 \end{equation}
 Since $ E\{\Delta_{\DTEone} s(X)\} = 0 $, the first term in \eqref{eq:der-delta} becomes

\[
E\left\{\frac{\alpha_{{\DTEone}}(X)}{\delta_{{\DTEone}}(X)} s(X)\right\} = E\left[\left\{\frac{\alpha_{{\DTEone}}(X)}{\delta_{{\DTEone}}(X)}   - \Delta_{\DTEone}\right\} s(X)  \right] .
\]
From Lemma \ref{lem:eif-D1D2Y1}, the second term in \eqref{eq:der-delta} is reduced to  
\begin{equation}
   \begin{aligned}
        E_t& \left\{\frac{\dot{\beta}^1_{{\DTEone},t}(X)-\dot{\beta}^0_{{\DTEone},t}(X) }{ \delta_{{\DTEone}}(X) } \right\}\bigg|_{t=0} \\
&= E\left(E\left[ \{{\psi_{Y_1D_2}^1-\psi_{Y_1D_2}^0 -\alpha_{\DTEone}^1(X)+\alpha_{\DTEone}^0(X) }\}s(Y_1,D_2 \mid Z_1, X) \mid X\right]  \dfrac{1}{{ \delta_{{\DTEone}}(X) }}\right)\\
&= E \left[\dfrac{\{{\psi_{Y_1D_2}^1-\psi_{Y_1D_2}^0 -\alpha_{\DTEone}^1(X)+\alpha_{\DTEone}^0(X) }\} }{{ \delta_{{\DTEone}}(X) }}  s(Y_1,D_2 \mid Z_1, X)  \right]\\
&= E \left[\dfrac{\{{\psi_{Y_1D_2}^1-\psi_{Y_1D_2}^0 -\alpha_{\DTEone } (X) }\} }{{ \delta_{{\DTEone}}(X) }}  s(Y_1,D_2 \mid Z_1, X)  \right]\\&= E\begin{pmatrix}
    \begin{bmatrix}
    \dfrac{Z_1\left\{ Y_1D_2 - E (Y_1D_2 \mid X, Z_1 = 1) \right\}}{\pr(Z_1 = 1 \mid X)} 
\\\addlinespace[1.5mm]- \dfrac{(1-Z_1)\left\{ Y_1D_2 - E (Y_1D_2 \mid X, Z_1=0 ) \right\}}{\pr(Z_1 = 0 \mid X)} \\ \addlinespace[1.5mm]-\alpha_{\DTEone } (X)
\end{bmatrix} \dfrac{s(Y_1,D_2 \mid Z_1, X)}{\delta_{{\DTEone}}(X) }
\end{pmatrix}\\&= E 
    \begin{bmatrix}
    \dfrac{(2Z_1-1)\left\{ Y_1D_2 - E (Y_1D_2 \mid X, Z_1 ) \right\}}{\pr(Z_1  \mid X)} \dfrac{1}{\delta_{{\DTEone}}(X) } s(Y_1,D_2 \mid Z_1, X)
\end{bmatrix},
   \end{aligned}
\end{equation}
where the last equality holds due to 
$$   E\left\{\frac{\alpha_{{\DTEone}}(X)}{\delta_{{\DTEone}}(X)}s(Y_1,D_2 \mid Z_1, X)\right\} =0.$$
Using Lemma \ref{lem:eif-D1D2Y1}, the third term in \eqref{eq:der-delta}  is reduced to:\begin{align*} 
E&\left[\frac{  \alpha_{{\DTEone}}(X) }{\delta^2_{{\DTEone}}(X)}\{\dot{\delta}^1_{\DTEone,t}(X)-\dot{\delta}^0_{\DTEone,t}(X)\}\right]\bigg|_{t=0} \\
&= E \left[\frac{  \alpha_{{\DTEone}}(X) }{\delta^2_{{\DTEone}}(X)}E\left[ \{{\psi_{D_1D_2}^1-\psi_{D_1D_2}^0 -\delta_{\DTEone}^1(X)+\delta_{\DTEone}^0(X) }\}s(D_1,D_2 \mid Z_1, X) \mid X\right]\right]\\
&= E \left[\frac{  \alpha_{{\DTEone}}(X) }{\delta^2_{{\DTEone}}(X)} \{{\psi_{D_1D_2}^1-\psi_{D_1D_2}^0 -\delta_{\DTE } (X) }\}s(D_1,D_2 \mid Z_1, X) \right]\\
&= E \left[\frac{  \alpha_{{\DTEone}}(X) }{\delta^2_{{\DTEone}}(X)}   \begin{bmatrix}
    \dfrac{Z_1\left\{ D_1D_2 - E (D_1D_2 \mid X, Z_1 = 1) \right\}}{\pr(Z_1 = 1 \mid X)} 
\\\addlinespace[1.5mm]- \dfrac{(1-Z_1)\left\{ D_1D_2 - E (D_1D_2 \mid X, Z_1=0 ) \right\}}{\pr(Z_1 = 0 \mid X)} \\ \addlinespace[1.5mm]-\delta_{\DTEone } (X)
\end{bmatrix}s(D_1,D_2 \mid Z_1, X) \right]\\
&=E 
    \begin{bmatrix}
    \dfrac{(2Z_1-1)\left\{ D_1D_2 - E (D_1D_2 \mid X, Z_1  ) \right\}}{\pr(Z_1  \mid X)} \dfrac{ {  \alpha_{{\DTEone}}(X) }}{{\delta^2_{{\DTEone}}(X)} } s(D_1,D_2 \mid Z_1, X)
\end{bmatrix}.
\end{align*}By plugging the three formulas into \eqref{eq:der-delta}, we obtain the following expression for $\left.  \dot{\Delta}_{t}\right|_{t=0}$:
\begin{equation}
    \label{eq:N-exp-deriv}
    \begin{aligned}
    \left.  \dot{\Delta}_{t}\right|_{t=0} &= E\left[\left\{\frac{\alpha_{{\DTEone}}(X)}{\delta_{{\DTEone}}(X)}   - \Delta_{\DTEone}\right\} s(X) \right]
\\
&~~~+ E\left[   \dfrac{(2Z_1-1)\left\{ Y_1D_2 - E (Y_1D_2 \mid X, Z_1 ) \right\}}{\pr(Z_1  \mid X)} \dfrac{1}{\delta_{{\DTEone}}(X) } s(Y_1,D_2 \mid Z_1, X)\right]
\\
&~~~- E\left[    \dfrac{(2Z_1-1)\left\{ D_1D_2 - E (D_1D_2 \mid X, Z_1  ) \right\}}{\pr(Z_1  \mid X)} \dfrac{ {  \alpha_{{\DTEone}}(X) }}{{\delta^2_{{\DTEone}}(X)} } s(D_1,D_2 \mid Z_1, X)\right].
\end{aligned}
\end{equation}We can verify that
\begin{gather*}
 \frac{\alpha_{{\DTEone}}(X)}{\delta_{{\DTEone}}(X)}   - \Delta_{\DTEone}\in H_1,\\
  \dfrac{(2Z_1-1)\left\{ Y_1D_2 - E (Y_1D_2 \mid X, Z_1 ) \right\}}{\pr(Z_1  \mid X)} \dfrac{1}{\delta_{{\DTEone}}(X) } \in H_4,\\
    \dfrac{(2Z_1-1)\left\{ D_1D_2 - E (D_1D_2 \mid X, Z_1  ) \right\}}{\pr(Z_1  \mid X)} \dfrac{{  \alpha_{{\DTEone}}(X) }}{{\delta^2_{{\DTEone}}(X)} }  \in H_4.
\end{gather*}
Since $H_1$, $H_2$, $H_3$, and $H_4$ are orthogonal to each other, we can rewrite \eqref{eq:N-exp-deriv} as:
\begin{equation*}
     \begin{aligned}
    \left.  \dot{\Delta}_{t}\right|_{t=0} &= E\left[\left\{\frac{\alpha_{{\DTEone}}(X)}{\delta_{{\DTEone}}(X)}   - \Delta_{\DTEone}\right\} s(O) \right]
\\
&~~~+ E\left[   \dfrac{(2Z_1-1)\left\{ Y_1D_2 - E (Y_1D_2 \mid X, Z_1 ) \right\}}{\pr(Z_1  \mid X)} \dfrac{1}{\delta_{{\DTEone}}(X) } s(O)\right]
\\
&~~~- E\left[    \dfrac{(2Z_1-1)\left\{ D_1D_2 - E (D_1D_2 \mid X, Z_1  ) \right\}}{\pr(Z_1  \mid X)} \dfrac{ {  \alpha_{{\DTEone}}(X) }}{{\delta^2_{{\DTEone}}(X)} } s(O)\right].
\end{aligned}
\end{equation*}
As a result, we obtain the EIF for $\Delta_\DTEone$:
\begin{align*}
    \varphi^{\eff}_{\DTEone}\left(O \right)&= \frac{\alpha_{{\DTEone}}(X)}{\delta_{{\DTEone}}(X)}   + \dfrac{(2Z_1-1)\left\{ Y_1D_2 - E (Y_1D_2 \mid X, Z_1 ) \right\}}{\pr(Z_1  \mid X)} \dfrac{1}{\delta_{{\DTEone}}(X) } \\
&~~~-  \dfrac{(2Z_1-1)\left\{ D_1D_2 - E (D_1D_2 \mid X, Z_1  ) \right\}}{\pr(Z_1  \mid X)} \dfrac{ {  \alpha_{{\DTEone}}(X) }}{{\delta^2_{{\DTEone}}(X)} } -\Delta_{\DTEone}.
\end{align*}
Simplifying the expression: 
\begin{equation*}
\begin{aligned}
   \varphi^{\eff}_{\DTEone}\left(O \right)&=\frac{2 Z_1-1}{\pr(Z_1  \mid X)}\left(\frac{Y_1D_2-E(Y_1D_2 \mid Z_1, X)}{\delta_{{\DTEone}}(X)}-\frac{(D_1D_2-E(D_1D_2 \mid Z_1, X))}{\delta_{{\DTEone}}(X)} \frac{\alpha_{{\DTEone}}(X)}{\delta_{{\DTEone}}(X)}\right)\\
&~~~+\frac{\alpha_{{\DTEone}}(X)}{\delta_{{\DTEone}}(X)} \\
& =\frac{2 Z_1-1}{\pr(Z_1  \mid X)} \frac{1}{\delta_{{\DTEone}}(X)}\{Y_1D_2-E(Y_1D_2 \mid Z_1, X)-(D_1D_2-E(D_1D_2 \mid Z_1, X)) \omega_{\DTEone}(X)\}\\
&~~~+\omega_{\DTEone}(X) -\Delta_{\DTEone} \\
& =\frac{2 Z_1-1}{\pr(Z_1  \mid X)} \frac{1}{\delta_{{\DTEone}}(X)}\left\{\begin{pmatrix}
  Y_1D_2-\omega_{\DTEone}(X) \delta_{{\DTEone}}(X) Z_1-\eta_{\DTEone}(X)\\-D_1D_2 \omega_{\DTEone}(X)+Z_1 \delta_{{\DTEone}}(X) \omega_{\DTEone}(X)+\mu_{\dte}(X) \omega_{\DTEone}(X)
  \end{pmatrix}\right\}\\
&~~~+\omega_{\DTEone}(X) -\Delta_{\DTEone} \\
& =\frac{2 Z_1-1}{\pr(Z_1  \mid X)} \frac{1}{\delta_{{\DTEone}}(X)}\left\{\begin{pmatrix}
    Y_1D_2-\eta_{\DTEone}(X)-D_1D_2 \omega_{\DTEone}(X)\\+\mu_{\dte}(X) \omega_{\DTEone}(X)
\end{pmatrix}\right\}
+\omega_{\DTEone}(X) -\Delta_{\DTEone},
\end{aligned}
\end{equation*} 
where we use the facts:
\begin{align*}
    E(Y_1D_2\mid Z_1,X)&=    E(Y_1D_2\mid Z_1=0,X)+Z_1\{  E(Y_1D_2\mid Z_1=1,X)-  E(Y_1D_2\mid Z_1=0,X)\}\\&=    E(Y_1D_2\mid Z_1=0,X)+Z_1\omega_{\DTEone}(X) \delta_{{\DTEone}}(X) \\&=    \eta_{\DTEone}(X)+Z_1\omega_{\DTEone}(X) \delta_{{\DTEone}}(X) ,\\ E(D_1D_2\mid Z_1,X)&=    E(D_1D_2\mid Z_1=0,X)+Z_1\{  E(D_1D_2\mid Z_1=1,X)-  E(D_1D_2\mid Z_1=0,X)\}\\ &=  \mu_{\dte}(X)+Z_1 \delta_{{\DTEone}}(X) .
\end{align*}
We show that  $\widehat{\Delta }_{\DTEone}$ with $ \{			 \widehat\pi_{\dte} (X) , \widehat \mu_{\dte}(X)  ,  {\widehat\eta}_{\DTEone} (X) , \widehat\omega_{\DTEone}(X) , {\widehat\delta}_{\DTEone} (X) \}$
     satisfying the regularity conditions (a)-(d) in Theorem \ref{THM: NON-MARCHINE-LEARNING} is asymptotically normal and has the influence function $ \varphi \left( O ;\Delta_{\DTEone}\right)$ in Theorem \ref{thm:eif-2}, therefore achieving the semiparametric efficiency. The proofs for other estimators are similar and hence omitted.  
     \section{The proof of Theorem \ref{THM: NON-MARCHINE-LEARNING}}
     \subsection{The proof of Theorem \ref{THM: NON-MARCHINE-LEARNING}(i)}
     \label{proff:thm3-1}
     Let $\theta$ denote the nuisance functions $\theta=\{			 \pi_{\dte} (X) ,  \mu_{\dte}(X)  ,  \eta_{\DTEone} (X) , \omega_{\DTEone}(X) , \delta_{\DTEone} (X) \}$ and $\widehat\theta$ be its corresponding estimator $\widehat\theta=\{			 \widehat\pi_{\dte} (X) , \widehat \mu_{\dte}(X)  ,  {\widehat\eta}_{\DTEone} (X) , \widehat\omega_{\DTEone}(X) , {\widehat\delta}_{\DTEone} (X) \}$.  Let $\theta^*$ be the probability limit of $ \widehat\theta$.     
      Let 
\begin{align*}
    \varphi(O;  \widehat\theta) &=\mathbb{P}_n\left[    \dfrac{(-1)^{1-Z_1}}{\widehat\pi_{\dte} (Z_1,X){\widehat\delta}_{\DTEone}(X)}\begin{Bmatrix}
 D_2Y_1-{\widehat\eta}_{\DTEone}(X)-D_1D_2\widehat\omega_{\DTEone}(X)\\+  \widehat\mu_{\dte}(X) \widehat\omega_{\DTEone}(X) 
    \end{Bmatrix}+\widehat\omega_{\DTEone}(X) \right].
\end{align*}

By the empirical process theory, we have
\begin{align}
 \notag   \P_n \{	\varphi(O;  \widehat\theta)\} - \P\{	 \varphi(O; \theta^*) \} & = (\P_n - \P) \varphi(O;  \widehat\theta) + \P \left\{ \varphi(O; \widehat \theta) - \varphi(O; \theta^*) \right\}\\ \notag  
&= (\P_n - \P) \varphi(O; \theta^*) + \P \left\{ \varphi(O; \widehat\theta) - \varphi(O; \theta^*) \right\} + o_{p}(1) , \notag   
\end{align}
where the first equality is followed by Assumption \ref{assump:regular}(a).  
We next analyze the  term $ \P \left\{ \varphi(O; \widehat\theta) - \varphi(O; \theta ) \right\}
$.  Following the derivation of the bias formula, we have
\begin{align*}
&\P\left[\frac{(-1)^{1-Z_1}}{\widehat\pi_{\dte}(Z_1,X){{\widehat\delta}_{\DTEone}(X)}}\left\{Y_1D_2-{{\widehat\eta}}_{\DTEone}(X)-D_1D_2{\widehat\omega}_{\DTEone}(X) + {\widehat\mu}_{\dte}(X) {\widehat\omega}_{\DTEone}(X)\right\}+{\widehat\omega}_{\DTEone}(X)-\omega_{\DTEone}\right]\\
&=\P\begin{bmatrix}\dfrac{(-1)^{1-Z_1}}{\widehat\pi_{\dte}(Z_1,X){{\widehat\delta}_{\DTEone}(X)}}\begin{Bmatrix}
\P\left(Y_1D_2\mid Z_1,X\right) - {{\widehat\eta}}_{\DTEone}(X) \\\addlinespace[1mm] -
  \P\left(D_1D_2\mid Z_1,X\right){\widehat\Delta}_{\DTEone} (X)\\\addlinespace[1mm]+{\widehat\mu}_{\dte}(X){\widehat\omega}_{\DTEone}(X)
\end{Bmatrix}+ {\widehat\omega}_{\DTEone}(X)-\omega_{\DTEone}(X)\end{bmatrix}\\
&=\P\begin{bmatrix}\dfrac{(-1)^{1-Z_1}}{\widehat\pi_{\dte}(Z_1,X) {{\widehat\delta}_{\DTEone}(X) }}\begin{Bmatrix}
\eta_{\DTEone}(X) +\delta_{{\DTEone}}(X) \omega_{\DTEone}(X) Z_1- {{\widehat\eta}}_{\DTEone}(X) \\ -\mu_{\dte}(X){\widehat\omega}_{\DTEone}(X)-\delta_{{\DTEone}}(X)Z_1{\widehat\omega}_{\DTEone}(X) \\  +{\widehat\mu}_{\dte}(X){\widehat\omega}_{\DTEone}(X)
\end{Bmatrix}+{\widehat\omega}_{\DTEone}(X) -\omega_{\DTEone}(X)\end{bmatrix}\\
&=\P\begin{pmatrix}\dfrac{(-1)^{1-Z_1}}{\widehat\pi_{\dte}(Z_1,X){{\widehat\delta}_{\DTEone}(X)}}\begin{bmatrix}
\eta_{\DTEone}(X)-{{\widehat\eta}}_{\DTEone}(X)\\+\left\{{\widehat\mu}_{\dte}(X)-\mu_{\dte}(X)\right\}{\widehat\omega}_{\DTEone}(X)\\+\delta_{{\DTEone}}(X)Z_1\{\omega_{\DTEone}(X)-{\widehat\omega}_{\DTEone}(X)\}
\end{bmatrix} +{\widehat\omega}_{\DTEone}(X)-\omega_{\DTEone}(X)\end{pmatrix} \\
&=\P\left[\frac{(-1)^{1-Z_1}\delta_{{\DTEone}}(X)Z_1}{\widehat\pi_{\dte}(Z_1,X){{\widehat\delta}_{\DTEone}(X)}}\left\{\omega_{\DTEone}(X)-{\widehat\omega}_{\DTEone}(X)\right\}-\left\{\omega_{\DTEone}(X)-{\widehat\omega}_{\DTEone}(X)\right\}\right]\\&\;\;\;\;\;\;\;+\P\begin{bmatrix}\left\{\dfrac{(-1)^{1-Z_1}}{\widehat\pi_{\dte}(Z_1,X){{\widehat\delta}_{\DTEone}}(X)}-\dfrac{(-1)^{1-Z_1}}{\pi_{\dte}(Z_1,X){{\widehat\delta}_{\DTEone}(X)}}\right\}\\\times\left[\left\{\eta_{\DTEone}(X)-{{\widehat\eta}}_{\DTEone}(X)\right\}+\left\{{\widehat\mu}_{\dte}(X)-\mu_{\dte}(X)\right\}{\widehat\omega}_{\DTEone}(X)\right]\end{bmatrix}\\
&=\P\left\{\frac{(-1)^{1-Z_1}\delta_{{\DTEone}}(X)Z_1}{\widehat\pi_{\dte}(Z_1,X){{\widehat\delta}_{\DTEone}}(X)}\left\{\omega_{\DTEone}(X)-{\widehat\omega}_{\DTEone}(X)\right\}-\left\{\omega_{\DTEone}(X)-{\widehat\omega}_{\DTEone}(X)\right\}\right\}\\&\;\;\;\;\;\;\;+\P\begin{bmatrix}\left\{\dfrac{(-1)^{1-Z_1}}{\widehat\pi_{\dte}(Z_1,X){{\widehat\delta}_{\DTEone}}(X)}-\dfrac{(-1)^{1-Z_1}}{\pi_{\dte}(Z_1,X){{\widehat\delta}_{\DTEone}}(X)}\right\}\\\times\left[\left\{\eta_{\DTEone}(X)-{{\widehat\eta}}_{\DTEone}(X)\right\}+\left\{{\widehat\mu}_{\dte}(X)-\mu_{\dte}(X)\right\}{\widehat\omega}_{\DTEone}(X)\right]\end{bmatrix}\\
&=\P\left[\left\{\frac{Z_1\delta_{{\DTEone}}(X)}{\widehat\pi_{\dte}(Z_1,X){{\widehat\delta}_{\DTEone}}(X)}-\frac{Z_1\delta_{{\DTEone}}(X)}{\pi_{\dte}(Z_1,X)\delta_{{\DTEone}}(X)}\right\}\times\left\{\omega_{\DTEone}(X)-{\widehat\omega}_{\DTEone}(X)\right\}\right]\\&\;\;\;\;\;\;\;+\P\begin{pmatrix}\left\{\dfrac{(-1)^{1-Z_1}}{\widehat\pi_{\dte}(Z_1,X){{\widehat\delta}_{\DTEone}}(X)}-\dfrac{(-1)^{1-Z_1}}{\pi_{\dte}(Z_1,X){{\widehat\delta}_{\DTEone}}(X)}\right\}\\\times\left[\left\{\eta_{\DTEone}(X)-{{\widehat\eta}}_{\DTEone}(X)\right\}+\left\{{\widehat\mu}_{\dte}(X)-\mu_{\dte}(X)\right\}{\widehat\omega}_{\DTEone}(X)\right]\end{pmatrix}\\
& =\P\left[\left(\frac1{\widehat\pi_{\dte}(Z_1,X){{\widehat\delta}_{\DTEone}(X)}}-\frac1{\pi_{\dte}(Z_1,X)\delta_{\DTEone}(X)}\right)\delta_{{\DTEone}} (X)Z_1\left\{{\omega}_{\DTEone}(X)-\widehat\omega_{\DTEone}(X)\right\}\right]\\&\;\;\;\;\;\;\;+\P\begin{pmatrix}\left\{\dfrac{(-1)^{1-Z_1}}{\widehat\pi_{\dte}(Z_1,X){{\widehat\delta}_{\DTEone}(X)}}-\dfrac{(-1)^{1-Z_1}}{\pi_{\dte}(Z_1,X){{\widehat\delta}_{\DTEone}(X)}}\right\}\\\times\left[\left\{\eta_{\DTEone}(X)-{{\widehat\eta}}_{\DTEone}(X)\right\}+\left\{{\widehat\mu}_{\dte}(X)-\mu_{\dte}(X)\right\}{\widehat\omega}_{\DTEone}(X)\right]\end{pmatrix}\\
&=\P\left\{\left(\frac{\pi_{\dte}(Z_1,X)\delta_{{\DTEone}}(X)-\widehat\pi_{\dte}(Z_1,X){{\widehat\delta}_{\DTEone}(X)}}{\widehat\pi_{\dte}(Z_1,X){{\widehat\delta}_{\DTEone}(X)}\pi_{\dte}(Z_1,X)\delta_{\DTEone}(X)}\right)\delta_{{\DTEone}}(X) Z_1\left\{\omega_{\DTEone}(X)-{\widehat\omega}_{\DTEone}(X)\right\}\right\}\\&\;\;\;\;\;\;\;+\P\begin{pmatrix}\left\{\dfrac{(-1)^{1-Z_1}}{\widehat\pi_{\dte}(Z_1,X){{\widehat\delta}_{\DTEone}(X)}}-\dfrac{(-1)^{1-Z_1}}{\pi_{\dte}(Z_1,X){{\widehat\delta}_{\DTEone}(X)}}\right\}\\\times\left[\left\{\eta_{\DTEone}(X)-{{\widehat\eta}}_{\DTEone}(X)\right\}+\left\{{\widehat\mu}_{\dte}(X)-\mu_{\dte}(X)\right\}{\widehat\omega}_{\DTEone}(X)\right]\end{pmatrix}\\
&=\P\left\{\frac{\delta_{{\DTEone}}(X)Z_1}{\widehat\pi_{\dte}(Z_1,X){{\widehat\delta}_{\DTEone}(X)}\pi_{\dte}(Z_1,X)\delta_{\DTEone}(X)}\begin{bmatrix}
    \left\{\pi_{\dte}(Z_1,X)\delta_{\DTEone}(X)-\widehat\pi_{\dte}(Z_1,X){{\widehat\delta}_{\DTEone}(X)}\right\} \\
    \times\left\{\omega_{\DTEone}(X)-{\widehat\omega}_{\DTEone}(X)\right\}
\end{bmatrix}\right\}\\&\;\;\;\;\;\;\;+\P\left(\frac{(-1)^{1-Z_1}}{{{\widehat\delta}_{\DTEone}(X)}}\frac{\pi_{\dte}(Z_1,X)-\widehat\pi_{\dte}(Z_1,X)}{\pi_{\dte}(Z_1,X)\widehat\pi_{\dte}(Z_1,X)}\times\left[\left\{\eta_{\DTEone}(X)-{{\widehat\eta}}_{\DTEone}(X)\right\}+\left\{{\widehat\mu}_{\dte}(X)-\mu_{\dte}(X)\right\}{\widehat\omega}_{\DTEone}(X)\right]\right), 
\end{align*}
where the fourth equality holds due to the fact
\begin{align*}  \P\begin{pmatrix} \dfrac{(-1)^{1-Z_1}}{\pi_{\dte}(Z_1,X){{\widehat\delta}_{\DTEone}(X)}} \left[\left\{\eta_{\DTEone}(X)-{{\widehat\eta}}_{\DTEone}(X)\right\}+\left\{{\widehat\mu}_{\dte}(X)-\mu_{\dte}(X)\right\}{\widehat\omega}_{\DTEone}(X)\right]\end{pmatrix}=0
 \end{align*}
the sixth equality hold due to 
$(-1)^{1-Z_1}\times Z_1=Z_1$.

By the Cauchy–Schwarz inequality and Assumption \ref{assump:regular}(a), (c), and (d), it follows that for some constant $C$, we have 
\begin{align*}
     \P_n&  \{	\varphi(O;  \widehat\theta)\} - \P\{	 \varphi(O; \theta^*) \}  \\&\leq  \begin{Bmatrix}
   C {\big\|\omega_{\DTEone}(X)-{\widehat\omega}_{\DTEone}(X)\big\|}_2{\big\|\pi_{\dte}(Z_1,X)\delta_{\DTEone}(X)-\widehat\pi_{\dte}(Z_1,X){{\widehat\delta}_{\DTEone}(X)}\big\|}_2\\\addlinespace[1mm]+C{\big\|\pi_{\dte}(Z_1,X)-\widehat\pi_{\dte}(Z_1,X)\big\|}_2 {\big\|\mu_{\dte}(X)-{\widehat\mu}_{\dte}(X)\big\|}_2 \\\addlinespace[1mm]+C{\big\|\pi_{\dte}(Z_1,X)-\widehat\pi_{\dte}(Z_1,X)\big\|}_2  {\big\|{{\widehat\eta}}_{\DTE}(X)-\eta_{\DTEone}(X)\big\|}_2 
\end{Bmatrix}.
\end{align*}
So we have \begin{align*}
          \P_n \{	\varphi(O;  \widehat\theta)\} - \P\{	 \varphi(O; \theta^*) \}   = 
          O_p\begin{pmatrix}\addlinespace[0.25mm]
    n^{-1/2} + \lVert\omega_{\DTEone}-{\widehat\omega}_{\DTEone}\rVert_2 \lVert\pi_{\dte}\delta_{\DTEone}-\widehat\pi_{\dte}{\widehat\delta}_{\DTEone}\rVert_2 \\\addlinespace[1mm]+ \lVert\pi_{\dte}-\widehat\pi_{\dte}\rVert_2 \lVert\mu_{\DTEone}-
{\widehat\mu}_{\DTEone}\rVert_2
+\lVert\pi_{\dte}-\widehat\pi_{\dte}\rVert_2 \lVert{\widehat\eta}_{\DTEone}-\eta_{\DTEone}\rVert_2 \\\addlinespace[0.25mm]
\end{pmatrix}.  
\end{align*}
By the empirical process theory, we have $ \P_n \{	\varphi(O;  \widehat\theta)\} - \P\{	 \varphi(O; \theta^*) \} =o_p(1)$. This completes the proof.
     \subsection{The proof of Theorem \ref{THM: NON-MARCHINE-LEARNING}(ii)}
     The previous proof follows a similar logic in Section \ref{proff:thm3-1}. 
  By the Cauchy–Schwarz inequality and Assumption \ref{assump:regular}(a), (b), and (c), it follows that for some constant $C$, we have 
\begin{align*}
     \P_n&  \{	\varphi(O;  \widehat\theta)\} - \P\{	 \varphi(O; \theta^*) \}  \\&\leq  \begin{Bmatrix}
   C {\big\|\omega_{\DTEone}(X)-{\widehat\omega}_{\DTEone}(X)\big\|}_2{\big\|\pi_{\dte}(Z_1,X)\delta_{\DTEone}(X)-\widehat\pi_{\dte}(Z_1,X){{\widehat\delta}_{\DTEone}(X)}\big\|}_2\\\addlinespace[1mm]+C{\big\|\pi_{\dte}(Z_1,X)-\widehat\pi_{\dte}(Z_1,X)\big\|}_2 {\big\|\mu_{\dte}(X)-{\widehat\mu}_{\dte}(X)\big\|}_2 \\\addlinespace[1mm]+C{\big\|\pi_{\dte}(Z_1,X)-\widehat\pi_{\dte}(Z_1,X)\big\|}_2  {\big\|{{\widehat\eta}}_{\DTE}(X)-\eta_{\DTEone}(X)\big\|}_2 
\end{Bmatrix} \\&=o_p(n^{-1/2}).
\end{align*}
Additionally imposing Assumption  \ref{assump:regular}(b) ,  we have \begin{align*} 
          \P_n \{	\varphi(O;  \widehat\theta)\} - \P\{	 \varphi(O; \theta^*) \}   = (\P_n - \P) \varphi(O; \theta^*)    + o_p(n^{-1/2}) = (\P_n - \P) \varphi(O; \theta)  + o_p(n^{-1/2}).
\end{align*}
 This completes the proof. 
\section{The proof of Lemma \ref{thm:solution-unique}}
 \begin{proof} 
    We first show that for any integrable function $v(X)$, $$E\{Z_1 \phi(1 , X) v(X)\}= E\{v(X)\}$$ holds if and only if $\psi_1(1 , X)=\pi^{-1}_{\dte}(1 , X)$. The sufficient part is obvious. We prove the necessary part. Let $v(X)=\exp \left(a^{\T} X \cdot i\right)$ be the test function, where $a \in \mathbb{R}^r$. By assumption,
\begin{align*}
\begin{aligned}
 E&\left\{\exp \left(a^{\T} X \cdot i\right)\right\}=E\left\{Z_1\psi(1 , X) \exp \left(a^{\T} X \cdot i\right)\right\}\\& = E\left[Z_1\left\{\psi_1(1 , X)-\pi^{-1}_{\dte}(1 , X)\right\} \exp \left(a^{\T} X \cdot i\right)\right]+E\left\{Z_1\pi^{-1}_{\dte}(1 , X) \exp \left(a^{\T} X \cdot i\right)\right\} .
\end{aligned}
\end{align*}
By the tower property of conditional expectation, $E\left\{Z_1\pi^{-1}_{\dte}(1 , X) \exp \left(a^{\T} X \cdot i\right)\right\}= E\left\{\exp \left(a^{\T} X \cdot i\right)\right\}$. Then,
\begin{align*}
0=E\left[Z_1\left\{\psi (1 , X)-\pi^{-1}_{\dte}(1 , X)\right\} \exp \left(a^{\T} X \cdot i\right)\right]=E\left[\{\pi_{\dte}(1 , X)\psi(1 , X)-1\} \exp \left(a^{\T} X \cdot i\right)\right]
\end{align*}
holds for all $a \in \mathbb{R}^r$. Due to the uniqueness of Fourier transform we can obtain $\pi_{\dte}(1 , X)\psi(1 , X)-1=0$, which implies $\psi (1 , X)=\pi^{-1}_{\dte}(1 , X)$ almost surely.    The results of $E\{(1-Z_1) w(0 , X) v(X)\}= E\{v(X)\}$ can be similarly established. 

Given $\psi (Z_1 , X)=\pi^{-1}_{\dte}(Z_1 , X)$, 
we next show that 
\begin{align*}
    E\{\phi(X) v(X)\}&=E\{ (-1)^{1-Z_1}D_1D_2 \psi(Z_1  ,X) v(X)\}\\&=E\{ (-1)^{1-Z_1}D_1D_2  \pi^{-1}_{\dte}(Z_1 , X)v(X)\}\\&=E\{ Z_1D_1D_2  \pi^{-1}_{\dte}(1 , X)v(X)-(1-Z_1) D_1D_2  \pi^{-1}_{\dte}(0 , X)v(X)\}
    \\&=E\left[\{ E(D_1D_2 \mid Z=1,X)  -E(D_1D_2 \mid Z=0,X)  \} v(X)\right]
    \\&=E \{ \delta_{\DTEone}(X)v(X)  \} .
\end{align*}Let $v(X)=\exp \left(a^{\T} X \cdot i\right)$ be the test function, where $a \in \mathbb{R}^r$. Then  we have,  \begin{align*}
0&=E\left[ \left\{ \delta_{\DTEone}(X)-\phi(X)\right\} \exp \left(a^{\T} X \cdot i\right)\right] .
\end{align*}Due to the uniqueness of Fourier transform we can obtain $ \delta_{\DTEone}(X)-\phi(X)=0$, which implies $ \delta_{\DTEone}(X)=\phi(X)$ almost surely. 
\end{proof} 

\section{Regular conditions for nonparmateric estimator}
To establish the large sample properties of $\widehat{\Delta}_{\DTEone}^\np$, we impose the following regular conditions:
\begin{condition}
\label{assumption:regular condition}
    We consider the following conditions:

 (a) $E\{\delta_{\DTEone}^{-1}(X)\}<\infty$ and $E\{\delta_{\DTEone}^{-1}(X)Y_1^2\}<\infty$.
 
  (b) The support $\mathcal{X}$ of the $r$-dimensional covariate $X$ is a Cartesian product of $r$ compact intervals.

 (c)  There exist two positive constants $\bar{C}$ and $\underline{C}$ such that $$ 
0<\underline{C} \leq \alpha_{\min }\left(E\left\{v_{\minK}(X) v_{\minK}^{\T}(X)\right\}\right) \leq \alpha_{\max }\left(E\left\{v_{\minK}(X) v_{\minK}^{\T}(X)\right\}\right) \leq \bar{C}<\infty,$$  
where $\alpha_{\max }\left(E\left\{v_{\minK}(X) v_{\minK}^{\T}(X)\right\}\right)$ and $ \alpha_{\min }\left(E\left\{v_{\minK}(X) v_{\minK}^{\T}(X)\right\} \right)$ denote the largest and smallest eigenvalues of $E\left\{v_{\minK}(X) v_{\minK}^{\T}(X)\right\}$, respectively.

 (d)  There exist three positive constants $\infty>\eta_1>\eta_2>1>$ $\eta_3>0$ such that $\eta_2 \leq \pi^{-1}_{\dte}(z_1 , x) \leq \eta_1$ and $-\eta_3 \leq \delta_{\DTEone}(x) \leq \eta_3, \forall(z_1, x) \in$ $\{0,1\} \times \mathcal{X}$.

 (e)  There exist $\alpha_{\minK}^{\star}, \beta_{\minK}^{\star}$, and $\gamma_{\minK}^{\star}$ in $\mathbb{R}^{\minK}$ and $\beta>0$ such that for any $z_1 \in\{0,1\}$,
\begin{gather*} 
\sup _{x \in \mathcal{X}} \mid \dot m^{-1} (\delta_{\DTEone}(x))-\gamma_{\minK}^{\star\T}v_{\minK}(x) \mid=O_p(K^{-\beta}),\\
\sup _{x \in \mathcal{X}} \mid \dot h^{-1}(\pi^{-1}_{\dte}(z_1,x))-z_1  \alpha_{\minK}^{\star\T} v_{\minK}(x)-(1-z_1)   \beta_{\minK}^{\star\T}v_{\minK}(x) \mid=O_p(K^{-\beta}) .
\end{gather*}

 (f) $\zeta(K)^4 K^3 / n \rightarrow 0$ and $\sqrt{n} K^{-\beta} \rightarrow 0$, where $\zeta(K)= \sup _{x \in \mathcal{X}}\left\|v_{\minK}(x)\right\|$ and $\|\cdot\|$ is the usual Frobenius norm defined by $\|A\|=$ $\sqrt{\operatorname{tr}\left(A A^{\T}\right)}$ for any matrix $A$.
\end{condition} 
\section{The proof of Proposition \ref{eq:asymp-var}}
The logic of this proof is entirely consistent with   Theorem 2 in \citet{ai2022simple}, and we omit it here.
\section{The proof of Proposition  \ref{thm:iden-ite}}
 
We now verify the results for Proposition  \ref{thm:iden-ite} (a),
\begin{gather*}
    E\bigg\{  \frac{(-1)^{1-Z_1} D_2 Y_1}{ \pi_{\ite}(Z_1,Z_2,X) \delta_{\ITEone}(Z_2,X) } \mid Z_2,X   \bigg\} = \frac{ E(D_2 Y_1\mid Z_2,X,Z_1=1) - E(D_2 Y_1 \mid Z_2,X,Z_1=0) }{ \delta_{\ITEone}(Z_2,X)},\\
    E\bigg\{  \frac{ (-1)^{1-Z_1} (1-D_2) Y_1 }{ \pi_{\ite}(Z_1,Z_2,X) \delta_{\ITEzero}(Z_2,X)} \mid Z_2,X \bigg\} = \frac{ E( (1-D_2) Y_1\mid Z_2,X,Z_1=1  ) - E((1-D_2) Y_1 \mid Z_2,X,Z_1=0) }{ \delta_{\ITEzero}(Z_2,X)},
\end{gather*}
which implies 
\begin{align*}
\omega_{\ITE}(Z_2,X) 
= &E\{ Y_1(1,d) -Y_1(0,d)\mid Z_2,X \}  \\
=  &E\bigg\{  \frac{(-1)^{1-Z_1} \mathbb{I}( D_2=d) Y_1}{ \pi_{\ite}(Z_1,Z_2,X) \delta_{\ITE}(Z_2,X) }  \mid Z_2,X   \bigg\},
\end{align*}
and Proposition  \ref{thm:iden-ite} (a) holds. Next, we verify Proposition  \ref{thm:iden-ite} (b),
\begin{align*}
   & E\bigg[   \frac{(-1)^{1-Z_1} }{\pi_{\ite}(Z_1,Z_2,X) } \mathbb{I}(D_2=d) \{ Y_1 - D_1 \omega_{\ITE}(Z_2,X) \}  \mid Z_2,X  \bigg] \\
  =&   E\{ \mathbb{I}( D_2=d) Y_1\mid Z_2,X,Z_1=1 \} - E\{ \mathbb{I}(D_2=d) Y_1 \mid Z_2,X,Z_1=0 \} - \delta_{\ITE}(Z_2,X) \omega_{\ITE}(Z_2,X)\\
  =&0.
\end{align*}
Finally, we directly verify Proposition  \ref{thm:iden-ite} (c),
\begin{align*}
 & E[  \mathbb{I}(D_2=d) \{ Y_1 - D_1 \omega_{\ITE}(Z_2,X) \} - \eta_{\ITE}(Z_2,X) + \mu_{\ITE}(Z_2,X) \omega_{\ITE}(Z_2,X) \mid Z_2,X,Z_1] \\
 =& \eta_{\ITE}(Z_2,X) + Z_1 \omega_{\ITE}(Z_2,X) \delta_{\ITE}(Z_2,X) - E\{ \mathbb{I}(D_2=d) D_1 \mid Z_2,X \} \omega_{\ITE}(Z_2,X) - \eta_{\ITE}(Z_2,X) \\
 &+ \mu_{\ITE}(Z_2,X) \omega_{\ITE}(Z_2,X) \\
 =&  Z_1 \omega_{\ITE}(Z_2,X) \delta_{\ITE}(Z_2,X)  + \mu_{\ITE}(Z_2,X) \omega_{\ITE}(Z_2,X) - \omega_{\ITE}(Z_2,X) (  E\{ \mathbb{I}(D_2=d) D_1 \mid Z_2,X,Z_1=0 \} \\
 &+ Z_1 [ E\{ \mathbb{I}(D_2=d) D_1 \mid Z_2,X,Z_1=1 \} - E\{ \mathbb{I}(D_2=d) D_1 \mid Z_2,X,Z_1=0 \}  ] )\\
 =&0,
\end{align*}
where we use the fact
\begin{align*}
    &E\{  \mathbb{I}(D_2=d) Y_1 \mid Z_2,X,Z_1 \} \\
    = &E\{  \mathbb{I}(D_2=d) Y_1 \mid Z_2,X,Z_1=0\} + Z_1 [ E\{  \mathbb{I}(D_2=d) Y_1 \mid Z_2,X,Z_1=1\} - E\{  \mathbb{I}(D_2=d) Y_1 \mid Z_2,X,Z_1=0\}  ] \\
    =& \eta_{\ITE}(Z_2,X) + Z_1 \delta_{\ITE}(Z_2,X) \omega_{\ITE}(Z_2,X),
    \end{align*}
and
    \begin{align*}
    &E\{  \mathbb{I}(D_2=d) D_1 \mid Z_2,X,Z_1 \} \\
    = &E\{  \mathbb{I}(D_2=d) D_1 \mid Z_2,X,Z_1=0\} + Z_1 [ E\{  \mathbb{I}(D_2=d) D_1 \mid Z_2,X,Z_1=1\} - E\{  \mathbb{I}(D_2=d) D_1 \mid Z_2,X,Z_1=0\}  ]\\
    = &\mu_{\ITE}(Z_2,X) + Z_1 \delta_{\ITE}(Z_2,X). 
\end{align*}

\section{The proof for identification of interaction effect}
The proof follows the similar logic in \eqref{eq:wald-estimand}.
 \section{Estimation procedure } 
\label{sec:par-est} 
For ease of reference, we rewrite some of the necessary notations below. 
 Let $\pi_{\dte}(Z_1,X)=\pr(Z_1\mid X)$,  $\mu_{\DTEone}(X)=E(D_{1}D_2\mid Z_1=0 ,X)$,   $\eta_{\DTEone}(X)=E(Y_1D_2\mid Z_1=0,X)$,  $\delta_{\DTEone}(X)= E(D_{1} D_2 \mid Z_{{1}}=1,X)-E(D_{1} D_2 \mid Z_{{1}}=0,X) $.  
Consider the working parametric models $$ \{\pi_{\dte}(Z_1, X; \xi_1),\mu_\DTEone(X; \xi_2),  \eta_\DTEone(X; \xi_3),  {\delta} _\dte( D_1, Z_1, X; \xi_4),   \omega_\DTEone(X; \xi_5)\} $$ indexed by finite-dimensional parameters $ \xi = (\xi_1^\T, \ldots, \xi_5^\T)^\T $.   Under Assumptions \ref{assump:iv} and \ref{assump:no-interaction}, a three-step procedure to estimate the nuisance parameters is as follows:

\begin{itemize}
    \item[]  {\it Step} 1: Solve the score equation 
  $ \mathbb{P}_n\{  G_1(O;  \xi_1 ) \}=0
   $ 
    to obtain $ \widehat \xi_1 $, where 
       \[
    G_1(O;  \xi_1 ) :=  
     \dfrac{\partial \log \pr( Z_1\mid    X;  \xi_1  ) }{\partial  \xi_1    }.  
    \] 

  \item[] {\it Step} 2: Solve the score equation 
  $ \mathbb{P}_n\{  G_2(O;  \xi_2 ) \}=0
   $ 
    to obtain $ \widehat{\xi}_2 $, where
       \[
    G_2(O;   \xi_2) :=  
     \dfrac{\partial \log   \pr( D_1=1, D_2=1\mid   Z_1=0, X;  \xi_2  )   }{\partial  {\xi}_2   }.   
    \] 
    
  \item[] {\it Step} 3: Solve  estimating  equation 
  $ \mathbb{P}_n\{  G_3(O;  \xi_3 ) \}=0
   $ 
    to obtain $ \widehat{\xi}_3 $, where
       \[
    G_3(O;   \xi_3) :=  
(1-Z_1)\{Y_1D_1 - \eta_{\DTEone}(X;\xi_3)\}
    \] 

 \item[]  {\it Step} 4:  Solve estimating equation $
   \mathbb{P}_n\{G_4(O; \widehat{\xi}_1,\widehat{\xi}_2, {\xi}_4)\}=0$ 
    to obtain $  \widehat{\xi}_4$, where
    \[
    G_4(O; \widehat{\xi}_1, \widehat{\xi}_1, {\xi}_4) := A_1(X) \left\{ D_1D_2 - \delta_{\DTEone}(X; \xi_4) Z_1 - \mu_\DTEone(X; \widehat\xi_2)\right\} \frac{ (- 1)^{1-Z_1}}{\pi_{\dte}(Z_1,X;\widehat\xi_1)} = 0,
    \]
where  $ A_1( X) $ is a user-speciﬁed vector function of the same dimension as $  \xi_5$.  
     \item[]  {\it Step} 5:  Solve estimating equation $
   \mathbb{P}_n\{G_5(O; \widehat{\xi}_1, \widehat{\xi}_2,\widehat{\xi}_3, {\xi}_5)\}=0$ 
    to obtain $  \widehat{\xi}_5$, where
    \[
    G_5(O; \widehat{\xi}_1,\widehat{\xi}_2, \widehat{\xi}_3, {\xi}_5) := A_2(X) \left\{ \begin{pmatrix}
        Y_1D_2 - \omega_{\DTEone}(X; \xi_5) D_1D_2 - \eta_\DTEone(X; \widehat\xi_3)\\\addlinespace[1mm] +\mu_\DTEone(X; \widehat\xi_2) \omega_{\DTEone}(X; \xi_5) 
    \end{pmatrix}\right\} \frac{ (- 1)^{1-Z_1}}{\pi_{\dte}(Z_1,X;\widehat\xi_1)} = 0,
    \] 
where  $ A_2( X) $ is a user-speciﬁed vector function of the same dimension as $  \xi_5 $.   
\item[]  {\it Step} 6: Derive the plug-in estimator $\widehat{\Delta}_{\DTEone}$ for $\Delta_{\DTEone}$: $$\widehat{\Delta}_{\DTEone} =\begin{array}{c} \mathbb{P}_n\left[\begin{array}{c}\dfrac{(-1)^{1-Z_1}}{\widehat\pi_{\dte} (Z_1,X;\widehat\xi_1){\widehat\delta}_{\DTEone}(X;\widehat\xi_4)}\begin{Bmatrix}
    D_2Y_1- {\widehat\eta}_{\DTEone}(X;\widehat\xi_3)\\-D_1D_2\widehat\omega_{\DTEone}(X;\widehat\xi_5)\\+  \widehat\mu_{\DTEone}(X;\widehat\xi_2) \widehat\omega_{\DTEone}(X;\widehat\xi_5)
\end{Bmatrix}+\widehat\omega_{\DTEone}(X;\widehat\xi_5) \end{array}\right] \end{array}.$$
\end{itemize}

\end{document}